%% file: column-outliers-full.tex
\documentclass[a4paper,UKenglish,cleveref, autoref, thm-restate]{lipics-v2021}
\usepackage{amsmath,amsfonts,amssymb,amsthm}
\usepackage{mathbbol}
\usepackage{amsthm}
\usepackage{graphicx,color}
\usepackage{boxedminipage}
\usepackage[linesnumbered, boxed, algosection,noline,noend]{algorithm2e}
\usepackage{framed}
\usepackage{thmtools}
\usepackage{xspace}
\usepackage{todonotes}
\usepackage{pgfplots}
\usepackage{xifthen}
\usepackage{tabularx}
\usepackage{capt-of}
\usepackage{thmtools}
\usepackage{subcaption}
\usepackage{booktabs}
\usepackage{calc}
\usepackage{dirtytalk}%Sayan
\usepackage{bm}%Sayan

\usetikzlibrary{calc}
\usetikzlibrary{shapes}
\usetikzlibrary{arrows.meta}
 
\newif\iffull
\fulltrue

\newcommand{\GF}{{\sf GF}$(2)$\xspace}
\newcommand{\rank}{{\rm rank}\xspace}

\newcommand{\hdist}{d_H}

\newcommand{\NPC}{\textsf{NP}-complete\xspace}
\DeclareMathOperator{\operatorClassP}{{\sf P}}
\newcommand{\classP}{\ensuremath{\operatorClassP}}
\DeclareMathOperator{\operatorClassNP}{{\sf NP}}
\newcommand{\classNP}{\ensuremath{\operatorClassNP}}

\DeclareMathOperator{\operatorClassFPT}{{\sf FPT}\xspace}
\newcommand{\classFPT}{\ensuremath{\operatorClassFPT}\xspace}
\DeclareMathOperator{\operatorClassW}{{\sf W}}
\newcommand{\classW}[1]{\ensuremath{\operatorClassW[#1]}}

\newcommand{\bfx}{\mathbf{x}}

\newcommand{\OO}{\Oh}	
\newcommand{\R}{\mathcal{R}}	
\newcommand{\Oh}{\mathcal{O}}

\newcommand{\yesinstance}{yes-instance\xspace}	
\newcommand{\noinstance}{no-instance\xspace}

\newcommand{\pname}{\textsc}
\newcommand{\ProblemFormat}[1]{\pname{#1}}
\newcommand{\ProblemIndex}[1]{\index{problem!\ProblemFormat{#1}}}
\newcommand{\ProblemName}[1]{\ProblemFormat{#1}\ProblemIndex{#1}{}\xspace}

\newcommand{\probIS}{\ProblemName{Independent Set}}

\newcommand{\probPVC}{\ProblemName{Partial Vertex Cover}}

\newcommand{\probProjClR}{\ProblemName{Robust Projective Clustering}}

\newcommand{\probBFact}{\ProblemName{Low Boolean-Rank Approximation}}



\newlength{\RoundedBoxWidth}
\newsavebox{\GrayRoundedBox}
\newenvironment{GrayBox}[1]%
   {\setlength{\RoundedBoxWidth}{.93\textwidth}
    \def\boxheading{#1}
    \begin{lrbox}{\GrayRoundedBox}
       \begin{minipage}{\RoundedBoxWidth}}%
   {   \end{minipage}
    \end{lrbox}
    \begin{center}
    \begin{tikzpicture}%
       \node(Text)[draw=black!20,fill=white,rounded corners,%
             inner sep=2ex,text width=\RoundedBoxWidth]%
             {\usebox{\GrayRoundedBox}};
        \coordinate(x) at (current bounding box.north west);
        \node [draw=white,rectangle,inner sep=3pt,anchor=north west,fill=white] 
        at ($(x)+(6pt,.75em)$) {\boxheading};
    \end{tikzpicture}
    \end{center}}     

\newenvironment{defproblemx}[2][]{\noindent\ignorespaces%
                                \FrameSep=6pt%
                                \parindent=0pt%
                \vspace*{-2em}
                \ifthenelse{\isempty{#1}}{%
                  \begin{GrayBox}{\textsc{#2}}%                
                }{%
                  \begin{GrayBox}{\textsc{#2} parameterized by~{#1}}%  
                }
                \begin{tabular*}{\textwidth}{@{\hspace{.1em}} >{\itshape} p{1.1cm} p{0.85\textwidth} @{}}%        
            }{
                \end{tabular*}%
                \vspace{-7mm}
                \end{GrayBox}%
                \ignorespacesafterend
            }

\newcommand{\defproblema}[3]{% FJR Version
  \begin{defproblemx}{#1}
    Input:  & #2 \\
    Task: & #3
  \end{defproblemx}
}%

 \newcommand{\probconstrainedcl}{\ProblemName{Constrained Clustering with Outliers}}%Sayan
\newcommand{\probClusteringWRO}{\ProblemName{Feature Selection}}
\newcommand{\probClusteringWCO}{\ProblemName{$k$-Clustering with Column Outliers}}%Sayan

\title{Parameterized Complexity of Feature Selection for Categorical Data Clustering}
\author{Sayan Bandyapadhyay}{Department of Informatics, University of Bergen, Norway}{Sayan.Bandyapadhyay@uib.no}{https://orcid.org/0000-0001-8875-0102}{}
\author{Fedor V. Fomin}{Department of Informatics, University of Bergen, Norway}{Fedor.Fomin@uib.no}{https://orcid.org/0000-0003-1955-4612}{}
\author{Petr A. Golovach}{Department of Informatics, University of Bergen, Norway}{Petr.Golovach@uib.no}{https://orcid.org/0000-0002-2619-2990}{}
\author{Kirill Simonov}{Algorithms and Complexity Group, TU Wien, Austria}{kirillsimonov@gmail.com}{https://orcid.org/0000-0001-9436-7310}{Supported by the Austrian Science Fund (FWF) via projects Y1329 (\emph{Parameterized Analysis in Artificial Intelligence}) and P31336 (\emph{New Frontiers for Parameterized Complexity}).}
\authorrunning{S. Bandyapadhyay, F. V. Fomin, P. A. Golovach, K. Simonov} % mandatory. First: Use abbreviated first/middle names. Second (only in severe cases): Use first author plus 'et al.'

\Copyright{Sayan Bandyapadhyay, Fedor V. Fomin, Petr A. Golovach, Kirill Simonov} % mandatory, please use full first names. LIPIcs license is "CC-BY";  http://creativecommons.org/licenses/by/3.0/


 % mandatory; please add comma-separated list of keywords

\relatedversion{}
\relatedversiondetails{Conference Version}{https://doi.org/10.4230/LIPIcs.MFCS.2021.14}

\funding{This work is supported by the Research Council of Norway via the project ``MULTIVAL''.}%optional, to capture a funding statement, which applies to all authors. Please enter author specific funding statements as fifth argument of the \author macro.

\acknowledgements{The authors are thankful to the anonymous reviewers for their helpful comments.}%optional

\nolinenumbers %uncomment to disable line numbering

\hideLIPIcs  %uncomment to remove references to LIPIcs series (logo, DOI, ...), e.g. when preparing a pre-final version to be uploaded to arXiv or another public repository

\ArticleNo{1}%%%%%%%%%%%%%%%%%%%%%%%%%%%%%%%%%%%%%%%%%%%%%%%%%%%%%%

\begin{document}
%\title{On Parameterized Complexity of Clustering With Outliers in Hamming Space}
\maketitle

%\thispagestyle{empty}

%\begin{abstract}
%    
%    
%    
%    
%\end{abstract}
%\tableofcontents

\begin{abstract}
We develop new algorithmic methods with provable guarantees for feature selection in regard to categorical data clustering. While feature selection is one of the most common approaches to reduce dimensionality in practice, most of the known feature selection methods are heuristics.  
We study the following mathematical model. We assume that there are some inadvertent (or undesirable) features of the input data that unnecessarily increase the cost of clustering. Consequently, we want to   select  a  subset of the original features from the data such that there is a small-cost clustering on the selected features. 
More precisely,  for given integers $\ell$ (the number of irrelevant features) and $k$ (the number of clusters), budget $B$, and a set of $n$ categorical data points (represented by $m$-dimensional vectors whose elements belong to a finite set of values $\Sigma$), we want to select $m-\ell$ relevant features such that the cost of any optimal $k$-clustering  on these features does not exceed $B$. Here the cost of a cluster is the sum of Hamming distances ($\ell_0$-distances) between the selected features of the elements of the cluster and its center.  The clustering cost is the total sum of the costs of the clusters.

 We use the framework of parameterized complexity to identify how the complexity of the problem depends on parameters $k$, $B$, and $|\Sigma|$. Our main result is an algorithm that solves the Feature Selection  problem in time $f(k,B,|\Sigma|)\cdot m^{g(k,|\Sigma|)}\cdot n^2$ for some functions $f$ and $g$. In other words, the problem is fixed-parameter tractable parameterized by $B$ when $|\Sigma|$ and   $k$ are constants. Our algorithm for Feature Selection is based on a solution to a more general problem, Constrained Clustering with Outliers. In this problem, we want to delete a certain number of outliers such that the remaining points could be clustered around centers satisfying specific constraints. One interesting fact about Constrained Clustering with Outliers is that besides Feature Selection, it encompasses many other fundamental problems regarding categorical data such as Robust Clustering, Binary and Boolean Low-rank Matrix Approximation with Outliers, and Binary Robust Projective Clustering.  Thus as a byproduct of our theorem, we obtain algorithms for all these problems. 
  We also complement our algorithmic findings with complexity lower bounds.
\end{abstract}

\input{introduction}

\input{overview-organization}

\input{algorithm-constrained-clustering}

\iffull

\input{algorithm-column-outliers}

\input{row-outliers}

\input{hardness}

\fi
\input{conclusion}

\bibliography{column-outliers}

\end{document}

%% file: introduction.tex
 %!TEX root = column-outliers.tex
%\newcommand{\probRLRMA}{\ProblemName{Robust Low GF(2)-Rank Approximation}}
\newcommand{\probRLRMA}{\ProblemName{Robust $\ell_0$-Low Rank Approximation}}
\newcommand{\probLRMA}{\ProblemName{$\ell_0$-Low Rank Approximation}}
\newcommand{\probRBRMA}{\ProblemName{Robust Low Boolean-Rank Approximation}}
 \newcommand{\probAtMostClust}{\ProblemName{Binary $k$-Clustering}}
 
%%%%%%%%%%%%%%
%\input{abstract.tex}
%\begin{abstract}
%We provide parameterized computational upper and lower bound for clustering and principal component analysis (PCA)  of categorical data in the presence of outliers. The goal is to cluster a set of categorical data points (represented by  vectors whose elements are taken from a limited set of values $\Sigma$) to minimize the total sum of the costs of clusters ignoring some  data points that can be labeled as outliers.   Here the cost of a cluster is the sum of Hamming distances ($\ell_0$-distances) from the elements of the cluster to its center. The goal of the PCA is to find a linear subspace of small dimension minimizing 
%    
%    
%\end{abstract}

%\newpage
%\clearpage
%\pagenumbering{arabic} 

\section{Introduction}\label{sec:intro}
Clustering is one of the most fundamental concepts in data mining and machine learning.
A considerable challenge to the clustering approaches is the high dimensionality of modern datasets. When the data contains many irrelevant features (or attributes), an application of cluster analysis with a complete set of features could significantly decrease the solution's quality. The typical approach to overcome this challenge in practice is \emph{feature selection}. The method is based on selecting a small subset of relevant features from the data and applying the clustering algorithm only on the selected features. The survey of %Alelyani et al.
  \cite{SalemAlelyani14} provides a comprehensive overview on methods for feature selection in clustering. 
 Due to the significance of feature selection, there is a multitude of heuristic methods addressing the problem. However, very few provably correct methods %for the $k$-means clustering objective 
 are known \cite{DBLP:conf/nips/BoutsidisMD09,BoutsidisZMD15,cohen2015dimensionality}.
 
 Kim et al.  \cite{DBLP:journals/ida/KimSM02} introduced a model of feature selection in the context of $k$-means clustering. We use their motivating example here. Decision-making based on market surveys is a pragmatic marketing strategy used by manufacturers to increase customer satisfaction. The respondents of a survey are segmented into similar-interest groups so that each group of customers can be treated in a similar way. Consider such a market survey data that typically contains responses of customers to a set of questions regarding their demographic and psychographic information, shopping experience, attitude towards new products and expectations from the business. The standard practice used by market managers to segment customers is to apply clustering techniques w.r.t. the whole set of features. However, depending on the application, responses corresponding to some of the features might not be relevant to find the target set of market segments. Also, some of the responses might contain incomplete or spurious information. To address this issue, Kim et al. \cite{DBLP:journals/ida/KimSM02} considered several quality criteria to return \textit{Pareto} optimal (or non-dominated) solutions that optimize one or more criteria. One such solution removes a suitable subset of features and cluster the data w.r.t. the remaining features.

The main objective of this work is to study clustering problems on  \textit{categorical} data. 
%In this paper, we are interested in studying clustering problems on \textit{categorical} data.
In statistics, a categorical variable is a variable that can admit  a fixed  number of possible values. 
For example, it could be a gender, blood type, political orientation,   etc. A prominent example of categorical data is binary data where the points are vectors each of whose coordinates can take value either 0 or 1. Binary data arise in several important applications. In electronic commerce, each transaction can be modeled as a binary vector (known as market basket data) each of whose coordinates denotes whether a particular item is purchased or not \cite{zhang2007binary,li2005general}. 
%In document clustering, each document can be modeled as a binary vector each of whose coordinates denote whether a specific word is present or not in the document \cite{zhang2007binary,li2005general}. 
The most common similarity  (or dissimilarity) measure for categorical data objects is the Hamming distance, which is basically the number of mismatched attributes of the objects. 

\section{Our results}
In this paper, we introduce a new model of feature selection w.r.t. categorical data clustering, which is motivated by the work of Kim et al. \cite{DBLP:journals/ida/KimSM02}. We assume that there are some inadvertent features of the input data that unnecessarily increase the cost of clustering. Consequently, in our model,    
we define the best subset of features (of a given size) as the subset that minimizes the corresponding cost of clustering. The goal is   to compute such a subset and the respective clusters. We provide the first parameterized algorithmic and complexity results for feature selection in regard to   categorical data clustering.
 %feature selection cluster analyses of categorical data in two situations: (i) a subset of features is to be removed, and (ii) a subset of data points (or outliers) is to be removed. To compare with the second case of clustering involving  traditional (column) outliers, in the first case, we refer to the irrelevant features as row outliers. Next, we define the clustering problem with row outliers. 
%We introduce the following model of categorical clustering with outliers.  

 Let  $\Sigma$ be  a finite set of non-negative integers. We refer to $\Sigma$ as the alphabet and we  denote the $m$-dimensional space over $\Sigma$ by  ${\Sigma}^m$.
  Given two $m$-dimensional vectors $\bm{x},\bm{y} \in {\Sigma}^m$,   the \emph{Hamming distance} (or \emph{$\ell_0$-distance})  $d_H(\bm{x},\bm{y})$ is the number of different coordinates in $\bm{x}$ and $\bm{y}$, that is
$
  d_H(\bm{x},\bm{y})=|\{i \in \{1,\dots, m\} : \bm{x}[i] \ne \bm{y}[i]\}|.
  $
 For a set of indices $S\subset \{1,2,\ldots,m\}$ and an $m\times n$ matrix $\bm{A}$, let $\bm{A}^{-S}$ be the matrix obtained from $\bm{A}$ by removing the rows with indices in $S$. We denote the columns of $\bm{A}^{-S}$ by $\bm{a}_{-S}^j$ for $1\le j\le n$.
%We want to identify $m-\ell$ features that are suited best for the clustering objective. 
%Let us note that the standard notation for the number of clusters in data mining and machine learning communities is $k$. However, since in parameterized complexity community $k$ is usually reserved for the main parameter, which in our case is the cost of the solution, we denote the number of clusters by $r$.   
%Even though the cost function in our problem is basically the $k$-median cost function, we call our problem by a different name. 
We consider the following mathematical model of feature selection.

\medskip
\defproblema{\probClusteringWRO}%
{An alphabet $\Sigma$, an $m\times n$ matrix $\bm{A}$ with columns $\bm{a}^1, \bm{a}^2,\ldots,\bm{a}^n$ such that $\bm{a}^j\in {\Sigma}^m$ for all $1\le j\le n$, a positive integer $k$, non-negative integers $B$ and $\ell$. }%
{Decide whether there is a subset $O \subset \{1, 2, \dots, m\}$ of size at most $\ell$, a partition $\{I_1,I_2,\ldots,I_k\}$ of $\{1,2,\ldots,n\}$, and vectors $\bm{c}_1,\bm{c}_2,\ldots,\bm{c}_k \in {\Sigma}^{m-\vert O\vert}$ such that \[\sum_{i=1}^k \sum_{j \in I_i} d_H(\bm{a}_{-O}^j,\bm{c}_i)\le B.\]}

%This is because categorical data is in general unordered and hence the median of a set of points is not defined in our case. 
 In the above definition and all subsequent problem definitions, without loss of generality, we assume that each cluster is non-empty, i.e.,  $I_i\ne \emptyset$  for each $1\le i\le k$. Note that, otherwise, one could probe different values $k' < k$ for the actual number of non-empty clusters. The problem is defined as a decision problem, however if the instance is a  \yesinstance we would also like to find such a clustering. 
For $\ell=0$, \probClusteringWRO is the popular \probAtMostClust problem, which is known to be 
  \classNP-hard for every  $k\geq 2$~\cite{Feige14b}. This makes it natural to investigate the parameterized complexity of \probClusteringWRO. 
 
 Our main result is the following theorem.    
%We prove the following theorem. 
%\begin{restatable}{theorem}{thmwro}\label{theorem:clustering_outliers}
%\probClusteringWRO  is solvable  in time $(kB)^{O(kB)} \vert \Sigma\vert^{kB}\cdot m^{O(k)}n^2$. 
%\end{restatable}
\begin{restatable}{theorem}{thmwro}\label{theorem:clustering_outliers}
\probClusteringWRO  is solvable  in time $f(k,B,\vert \Sigma\vert)\cdot m^{g(k,\vert \Sigma\vert)}\cdot n^2$, where $f$ and $g$ are computable functions. 
\end{restatable}
 
In particular, this implies that for fixed $k$ and $|\Sigma|$, the problem is fixed-parameter tractable %or \classFPT
(\classFPT\footnote{A problem is \classFPT or fixed-parameter tractable parameterized by a set of parameters if it can be solved by algorithms that are exponential only in the values of the parameters while polynomial in the size of the input.}) 
parameterized by $B$. Note that in any study concerning the parameterized complexity of a problem, the value of the  parameter is implicitly assumed to be sufficiently small compared to the input size.  
Although the parameter $B$ seems to be a natural choice from the problem definition, in general $B$ can be fairly large. Hence, Theorem~\ref{theorem:clustering_outliers} is mostly   applicable in the scenario when for the selected features the cost  of clustering $B$ is small and thus the points are well-clustered on the selected features. Even in this case the problem is far from being trivial.  %er compared to the input-size might seem very restricted. 
%The main reason that motivates the choice of the parameter $B$ is that  
 One can think of our problem as a problem from the broader class of editing problems, where the goal is to check whether a given instance is close to a ``structured'' one. In particular, our problem can be seen as the problem of editing the input matrix after removing at most $\ell$ rows such that the resulting matrix contains at most $k$ distinct columns and the number of edits does not exceed the budget $B$. In this sense, our work is in line with the work  of 
 \cite{GanianKOS20} on matrix completion and \cite{FominGP20} on clustering. Note that in many applications it is reasonable to assume that $k$ and $|\Sigma|$ are bounded, as the alphabet size and the number of clusters are small in practice. Indeed, for binary data, $|\Sigma|=2$.  
 
  %of Marx \cite{Marx08} and Boucher et al. \cite{BoucherLL15} on Consensus Patterns and the works of Fomin et al. in \cite{FominGP20} and \cite{fomin_et_al:LIPIcs:2019:11576} on related clustering problems. The algorithms designed in all of these works are highly non-trivial, which shows that these problems remain sufficiently challenging even if the number of edits is assumed to be small. 

Another interesting property of our algorithm is that the running time does not depend on the number of irrelevant features $\ell$. In particular, for fixed $k$, $B$, and $|\Sigma|$, it runs in polynomial time even when $\ell=\Omega(m)$. Also, the theorem could be used to identify the minimum number of irrelevant features $\ell$ such that the cost of $k$-clustering on the remaining features does not exceed $B$. Note that our time complexity also exponentially depends on the number of clusters $k$. In this regard, one can compare our result with the result in \cite{FominGP20} that shows that the binary version of the problem with $\ell=0$  (\probAtMostClust) is \classFPT parameterized by $B$ only. However, in the presence of irrelevant features, the dependence on $k$ is unavoidable as we state in our next theorem. 
%The proof of the following theorem and preliminaries on  $\classW1$-hardness and ETH are in the attached full version.

\begin{restatable}{theorem}{thmhardness}
  \probClusteringWRO is $\classW1$-hard parameterized by 
  \begin{itemize}
\item  either $k + (m - \ell)$
\item or $\ell$
  \end{itemize}
 even when $B=0$ and $\Sigma=\{0,1\}$. Moreover, assuming the Exponential Time Hypothesis (ETH), the problem cannot be solved in time  $f(k)\cdot m^{o(k)}\cdot n^{\Oh(1)}$ for any function $f$, even when 
 $B = 0$ and the alphabet $\Sigma$ is binary.
    \label{thm:w1r}
\end{restatable}
%
%\begin{theorem}
%    \probClusteringWRO is W[1]-hard when parameterized by $\ell$ when $k = 0$. Assuming ETH, there is no algorithm solving the problem with $k = 0$ in time $n^{o(\ell)}$.
%    \label{thm:w1l}
%\end{theorem}
%

Note that when $B=0$ and $\Sigma=\{0,1\}$, from Theorem \ref{theorem:clustering_outliers} it follows that \probClusteringWRO   can be solved in time $f(k)\cdot m^{g(k)}\cdot n^2$.  Theorem \ref{thm:w1r} shows that the dependence of such a  function $g$ on $k$ is inevitable, unless $\classW{1}=\classFPT$, and $g(k)$ is unlikely to  be sublinear up to ETH. 
%Thus, in this case the upper bound matches with the lower bound in Theorem \ref{thm:w1r} based on ETH. 

In order to prove Theorem~\ref{theorem:clustering_outliers}, we prove 
%
 %follows (by the reduction given in Lemma~\ref{lem:ROsplcase} in the Appendix) from a 
 a more general theorem about 
\probconstrainedcl. In this problem, one seeks a clustering with centers of clusters satisfying the property imposed by a set of relations.
 Constrained clustering  \cite{FominGLP020} was introduced  as the tool in the design of approximation algorithms for binary low-rank approximation problems. The \probconstrainedcl problem is 
  basically the robust variant of this problem. 
As we will see,  by the reduction given in Lemma~\ref{lem:ROsplcase},  Theorem~\ref{theorem:main} proves Theorem~\ref{theorem:clustering_outliers}.  
Besides  \probClusteringWRO, \probconstrainedcl
encompasses a number of well-studied problems  around robust clustering, low-rank matrix approximation, and dimensionality reduction. Our algorithm for constrained clustering implies fixed-parameter tractability  for all these problems.

To define constrained clustering, we need a few definitions. 
 A \emph{$p$-ary relation} on $\Sigma$ is a collection of $p$-tuples whose elements are in $\Sigma$. 
 %For example, $R=\{(0,1,0),(1,1,0)\}$ is a $3$-ary relation on $\{0,1\}$. %\todo{A vector space is defined specifically over a field} 

\begin{definition}[\textbf{Vectors satisfying $\R$}] An ordered  set $C=\{\bm{c}_1,\bm{c}_2,\ldots,\bm{c}_p\}$ of $m$-dimensional vectors in ${\Sigma}^m$ is said to satisfy a set $\R=\{R_1,R_2,\ldots,R_m\}$ of $p$-ary relations on ${\Sigma}$ if for all $1\le i\le m$, 
the $p$-tuple formed by the $i$-th coordinates of vectors from $C$, that is $(\bm{c}_1[i],\bm{c}_2[i],\ldots,\bm{c}_p[i])$,  belongs to $R_i$. 
\end{definition}

%For example, consider  $\Sigma=\{0,1\}$, $m=2$, $p=3$,  and relation $\R=\{R_1, R_2\}$, where $R_1=\{(0,1,0),(1,1,0)\}$, $R_2=\{(1,0,0),(1,1,1)\}$. Then the set with 
% \[
% \bm{c}_1 = \begin{pmatrix}
%            0\\
%            1 \\
%         \end{pmatrix},
% \bm{c}_2 = \begin{pmatrix}
%            1\\
%            0 \\
%         \end{pmatrix},
% \bm{c}_3 = \begin{pmatrix}
%            0\\
%            0 \\
%         \end{pmatrix}
% \]
%  satisfies $\R$, as $(\bm{c}_1[1],\bm{c}_2[1],\bm{c}_3[1])=(0,1,0)\in R_1$ and $(\bm{c}_1[2],\bm{c}_2[2],\bm{c}_3[2])=(1,0,0)\in R_2$. 
%  On the other hand, the set with  
%   \[
% \bm{c}_1 = \begin{pmatrix}
%            0\\
%            1 \\
%         \end{pmatrix},
% \bm{c}_2 = \begin{pmatrix}
%            0\\
%            0 \\
%         \end{pmatrix},
% \bm{c}_3 = \begin{pmatrix}
%            0\\
%            0 \\
%         \end{pmatrix}
% \]
% does not   satisfy $\R$, because  $(\bm{c}_1[1],\bm{c}_2[1],\bm{c}_3[1])=(0,0,0)\not\in R_1$.

  We define the following constrained variant of robust categorical clustering.   %\vspace*{-1.5em}

%  Given two $m$-dimensional vectors $\bm{x},\bm{y} \in {\Sigma}^m$, we denote the \textit{Hamming distance} between $\bm{x}$ and $\bm{y}$, that is the number  $i\in \{1,\dots,m\}$ such that $\bm{x}[i]\neq \bm{y}[i]$, by $d_H(\bm{x},\bm{y})$. We define the following constrained variant of robust categorical clustering. 
  
%  Given two $m$-dimensional vectors $\bm{x},\bm{y} \in {\Sigma}^m$, the \textit{Hamming distance} between $\bm{x}$ and $\bm{y}$ is $d_H(\bm{x},\bm{y})=\sum_{i=1}^m$ equality$(\bm{x}[i],\bm{y}[i])$, where equality$(\bm{x}[i],\bm{y}[i])=0$ if $\bm{x}[i]= \bm{y}[i]$ and equality$(\bm{x}[i],\bm{y}[i])=1$ otherwise.
  %\\
  %\vspace{2mm}
  \medskip
\defproblema{\probconstrainedcl}
{An alphabet $\Sigma$, an $m\times n$ matrix $\bm{A}$ with columns $\bm{a}^1, \bm{a}^2,\ldots,\bm{a}^n$ such that $\bm{a}^j\in {\Sigma}^m$ for all $1\le j\le n$, a positive integer $k$, non-negative integers $B$ and $\ell$, a set $\R=\{R_1,R_2,\ldots,R_m\}$ of $k$-ary relations on $\Sigma$.}%
{Decide whether there is a subset $O \subset \{1, 2, \dots, n\}$ of size at most $\ell$, a partition $\mathcal{I}=\{I_1,I_2,\ldots,I_k\}$ of $\{1,2,\ldots,n\} \setminus O$, and a set $C=\{\bm{c}_1,\bm{c}_2,\ldots,\bm{c}_k\}$ of $m$-dimensional vectors in ${\Sigma}^m$ such that $C$ satisfies $\R$ and \[\sum_{i=1}^k \sum_{j \in I_i} d_H(\bm{a}^j,\bm{c}_i)\le B.\]}

Thus in \probconstrainedcl we want to identify a set of outliers $\bm{a}^i, i\in O$, such that the remaining 
  $n-\ell$  vectors could be  partitioned into $k$ \emph{clusters} $\{I_1,I_2,\ldots,I_k\}$. Each cluster $I_j$ could be  identified by its \emph{center} 
  $\bm{c}_j\in {\Sigma}^m$ as the set of vectors that are closer to   $\bm{c}_j\in {\Sigma}^m$ than to any other center (ties are broken arbitrarily).  Then the cost of 
 each cluster $I_j$ is  the  sum of the Hamming distances between its vectors and the corresponding center $\bm{c}_j\in {\Sigma}^m$.
 However, there is an additional condition that the set of cluster centers $C=\{\bm{c}_1,\bm{c}_2,\ldots,\bm{c}_k\}$ must satisfy the set of $k$-ary relations $\R$. And, the  total sum of costs of all clusters must not exceed $B$.
 \iffalse
 %\probClusteringWCO is the special case of \probconstrainedcl corresponding to the case when  every  relation $R_i\in \R$ contains   all possible $r$-tuples over $\Sigma$, that is, there are no constraints on the centers.
 
 %The  \probconstrainedcl problem is a robust variant of the \textsc{Binary Constrained Clustering} defined in 
%\cite{FominGLP020}. In fact, \textsc{Binary Constrained Clustering}  is the special case of the problem with $\Sigma=\{0,1\}$ and $\ell=0$.

 %\todo[inline]{This part should go somewhere else:
%We denote an input instance of the problem by $(\bm{A},r,k,\ell,\R)$. The desired partition is referred to as a \textit{clustering} of the instance, and each part $I_i$ is called a \emph{cluster}. 

%We denote an input instance of \probClusteringWCO by $(\bm{A},r,k,l)$. 
 %Note that \probClusteringWCO is a special case of \probconstrainedcl where for each $1\le i\le m$, $R_i$ contains all possible $r$-tuples on $\Sigma$. 

%   We denote an input instance of the problem by $(\bm{A},r,k,l)$. 
  %  Later, we will show that \probClusteringWRO is also a special case of \probconstrainedcl. 
%}

%\todo{we are using these notations without definiing in the applications part -- should we just remove them? or define before using}
\fi
We prove the  following theorem. 
  \begin{restatable}{theorem}{thmmain}\label{theorem:main}
      \probconstrainedcl is solvable  in time
      \vspace{-2pt}
      \[(kB)^{O(kB)} {\vert \Sigma\vert}^{kB}\cdot n^{O(k)}\cdot m^2.\]
\end{restatable}

      \vspace{-2pt}

The connection between \probconstrainedcl and \probClusteringWRO is established in the following lemma. 

 \begin{restatable}{lemma}{lemROsplcase}\label{lem:ROsplcase}
 For any instance $(\bm{D},k,B,\ell)$ of \probClusteringWRO, one can construct in time $\OO(mn+k\cdot |\Sigma|^k)$ an equivalent instance $(\bm{A},k',B',\ell',\R)$ of \probconstrainedcl such that $\bm{A}$ is the transpose of $\bm{D}$, $k'={\vert\Sigma\vert}^{k}$, $B'=B$ and $\ell'=\ell$. 
 \end{restatable}
% The proof of the lemma is in Technical Appendix. 
 % \todo[inline]{mention where the proof will be}
  
 Theorem~\ref{theorem:clustering_outliers}  follows from Theorem~\ref{theorem:main}   and Lemma~\ref{lem:ROsplcase}.  Connections of constrained clustering with several other clustering and low-rank matrix approximation problems have been established in the literature \cite{FominGLP020}. Similarly,  Theorem~\ref{theorem:main} allows to design parameterized algorithms for robust variants of these  problems.

 \medskip\noindent\textbf{Robust low-rank matrix approximation.}  Here we discuss two problems where for a given matrix of categorical data, we seek to remove $\ell$ columns such that the remaining columns are well approximated by a matrix of small rank.  
  The vanilla case of the \textsc{$\ell_0$-Low Rank Approximation} problem is the following.  Given an    $m\times n$ matrix $\bm{A}$ over  $\operatorname{GF}(p)$ (a finite field of order $p$),   
the task is to find an    $m\times n$ matrix $\bm{B}$  over  $\operatorname{GF}(p)$ of GF($p$)-rank at most $r$  which  is  closest to  $\bm{A}$ in the \emph{$\ell_0$-norm}, i.e., the goal is to minimize
$\|\bm{A}-\bm{B}\|_0$, the number of different entries in $\bm{A}$ and $\bm{B}$.
In the robust version of this problem, some of the columns of $\bm{A}$ could be outliers, which brings us to the following problem.
%\vspace*{-1.0em}

\medskip
 \defproblema{\probRLRMA}%
{An   $m\times n$ matrix $\bm{A}$ over  $\operatorname{GF}(p)$, a positive integer $r$, non-negative integers $B$ and $\ell$. }%
{Decide whether there is a  matrix $\bm{B}$ of GF(p)-rank at most $r$, and a  matrix $\bm{C}$  over  $\operatorname{GF}(p)$ with at most $\ell$ non-zero columns  such that $\|\bm{A}-\bm{B} -\bm{C} \|_0\leq B$.

}

Note that in this definition the non-zero columns of $\bm{C}$ can take any values. However, it is easy to see that the problem would be equivalent if the columns of $\bm{C}$ were constrained to be either zero columns or the respective columns of $\bm{A}$. This holds since if $\bm{C}$ contains a non-zero column, it could be replaced by the respective column of $\bm{A}$, and the respective column of $\bm{B}$ can be replaced by a zero column. This does not increase the cost nor the rank of $\bm{B}$. Thus any of the two formulations allows to restore the column outliers in $\bm{A}$ from $\bm{C}$.

%The following lemma gives a  reduction from \probconstrainedcl.
%
%
%\begin{lemma}\label{lem:BmatrixFasRClust} For any instance $(\bm{A},r, B,\ell)$ of  \probRLRMA, 
%one can  construct in time $\OO(m+n+p^{\Oh(r)})$ an equivalent instance 
%$(\bm{X},k',B',\ell,\R)$ of \probconstrainedcl, where  $\bm{X}=\bm{A}$,
%$k'=p^r$, $B'=B$ and $\R$ is a set of constraints.
%\end{lemma}
%(By equivalent instances we mean that   
% $(\bm{A},r, B,\ell)$ is a \yesinstance if and only if $(\bm{X},k',B',\ell,\R)$  is a \yesinstance and solution corresponding to a \yesinstance $(\bm{X},k',B',\ell,\R)$ can be used to construct a solution for the  \yesinstance $(\bm{A},r, B,\ell)$.)
%The proof of Lemma~\ref{lem:BmatrixFasRClust} is almost identical to the proof of  \cite[Lemma~1]{FominGLP020} and we omit it here. 

By a reduction \cite[Lemma~1]{FominGLP020} similar to  Lemma~\ref{lem:ROsplcase}, we can show that 
Theorem~\ref{theorem:main} yields the following theorem. 
\begin{theorem}\label{cor:RLRMA}
\probRLRMA  is FPT parameterized by $B$ when $p$ and $r$ are constants. 
\end{theorem}
%\todo[inline]{check running times. And maybe  even better to say FPT without specifying it}
% Let us note that for a binary matrix $\bfS$,  the square of the Frobenius norm   $\|\bfS\|_F^2$  is equal to the number of non-zero elements in $\bfS$, that is 
%  $\| S\| _{0}$. Corollary~\ref{cor:RLRMA} is easily extendable to the case when matrices $\bm{A}, \bm{B}$, and $\bm{C}$  are over finite field GF(p), and GF(p)-rank of $\bm{B}$ does not exceed $r$. In this case deciding whether $\|\bm{A}-\bm{B} -\bm{C} \|_0\leq k$,  that is, solving 
%  \textsc{Robust Low GF(p)-Rank Approximation} can be done in time  $(k^{\Oh(kr2^r)}n^{\Oh(r2^r)}m^2{ p}^{rk})^r$.
% 
 
% \subparagraph*{Robust Low Boolean-Rank Approximation} is a variant of the robust  low-rank approximation problem when the approximation matrix $\bm{B}$ is of low Boolean rank. 
 Another popular variant of   low-rank matrix approximation is the case when  the approximation matrix $\bm{B}$ is of low Boolean rank. 
More precisely, 
 let $\bm{A}$ be a binary $m\times n$ matrix. Now we consider the elements of $\bm{A}$ to be \emph{Boolean} variables. 
The \emph{Boolean rank} of $\bm{A}$ is the minimum $r$ such that $\bm{A}=\bm{U}\wedge \bm{V}$ for a Boolean $m\times r$ matrix $\bm{U}$ and a Boolean $r\times n$ matrix $\bm{V}$, where the product is Boolean, that is,  the logical $\wedge$ plays the role of multiplication and $\vee$ the role of sum.
 \iffalse
 Here  $0\wedge 0=0$, $0 \wedge 1=0$, $1\wedge 1=1$ , $0\vee0=0$, $0\vee1=1$, and  $1\vee 1=1$.  
Thus the  matrix product is over the Boolean semi-ring $({0, 1}, \wedge, \vee)$. This can be equivalently expressed
as the  normal matrix product with addition defined as $1 + 1 =1$. Binary matrices equipped with such algebra are called \emph{Boolean
matrices}. 
\fi
The variant of the low Boolean-rank matrix approximation is the following problem. 
 
 %\vspace*{-0.5em}
\medskip
 \defproblema{\probRBRMA}%
{A binary  $m\times n$ matrix $\bm{A}$, a positive integer $r$, non-negative integers $B$ and $\ell$. }%
{Decide whether there is a binary matrix $\bm{B}$ of Boolean rank  $\leq r$, and a binary matrix $\bm{C}$ with at most $\ell$ non-zero columns  such that $\|\bm{A}-\bm{B} -\bm{C} \|_0\leq B $.
    %\vspace*{-0.5em}

}

%\todo{columns of C are from $\bm{A}$ ?}

By Theorem~\ref{theorem:main} and reduction from constrained clustering to Boolean-rank matrix approximation identical to  \cite[Lemma~2]{FominGLP020},  we have the following.
\begin{theorem}\label{cor:BLRMA}
\probRBRMA  is FPT parameterized by $B$ when $r$ is a constant. 
\end{theorem}

Finally, we consider  clustering with outliers. This problem looks very similar to feature selection.
The only difference is that instead of features (the rows of the matrix $\bm{A}$), we seek to remove some columns of  $\bm{A}$. More precisely, we consider the following problem.   %\vspace*{-0.5em}

\medskip
\defproblema{\probClusteringWCO}%
{An alphabet $\Sigma$, an $m\times n$ matrix $\bm{A}$ with columns $\bm{a}^1, \bm{a}^2,\ldots,\bm{a}^n$ such that $\bm{a}^j\in {\Sigma}^m$ for all $1\le j\le n$, a positive integer $k$, non-negative integers $B$ and $\ell$. }%
{Decide whether there is a subset $O \subset \{1, 2, \dots, n\}$ of size at most $\ell$, a partition of $\{1,2,\ldots,n\} \setminus O$ 
into $k$ sets $\{I_1,I_2,\ldots,I_k\}$ called clusters, and vectors $\bm{c}_1,\bm{c}_2,\ldots,\bm{c}_k \in {\Sigma}^m$ such that  the cost of clustering \[\sum_{i=1}^k \sum_{j \in I_i} d_H(\bm{a}^j,\bm{c}_i)\le B.\]}

%In particular, this implies that for fixed $r$ the problem is fixed-parameter tractable (\classFPT) parameterized by $k$ and $|\Sigma|$. 
%\todo[inline]{add lower bounds}

%As mentioned before, we also consider the problem which is very similar to \probClusteringWRO, except here we have column outliers instead of the row outliers.   

%\defproblema{\probClusteringWCO}%
%{An alphabet $\Sigma$, an $m\times n$ matrix $\bm{A}$ with columns $\bm{a}^1, \bm{a}^2,\ldots,\bm{a}^n$ such that $\bm{a}^j\in {\Sigma}^m$ for all $1\le j\le n$, a positive integer $k$, non-negative integers $B$ and $\ell$. }%
%{Decide whether there is a subset $O \subset \{1, 2, \dots, n\}$ of size at most $\ell$, a partition of $\{1,2,\ldots,n\} \setminus O$ 
%into $k$ sets $\{I_1,I_2,\ldots,I_k\}$ called clusters, and vectors $\bm{c}_1,\bm{c}_2,\ldots,\bm{c}_k \in {\Sigma}^m$ such that  the cost of clustering $$\sum_{i=1}^k \sum_{j \in I_i} d_H(\bm{a}^j,\bm{c}_i)\le B.$$}
 
Note that \probClusteringWCO is also  a special case of \probconstrainedcl when  every  relation $R_i\in \R$ contains   all possible $k$-tuples over $\Sigma$, that is, there are no constraints on the centers. Hence, by Theorem \ref{theorem:main}, we readily obtain the same result for this problem. However, in this special case we show that it is possible to obtain an improved result.
% that matches with the result without outliers.  
 %\todo[inline]{a bit discussion of the problem}

%\todo[inline]{XXX}

%We prove the following theorem. 
\begin{restatable}{theorem}{thmwco}\label{theorem:WCO}
\probClusteringWCO is solvable  in time $2^{O(B\log B)}{\vert\Sigma\vert}^{B}\cdot(nm)^{O(1)}$.
\end{restatable}

%Here we show that \probClusteringWRO is a special case of \probconstrainedcl. 

%Note that for any matrix, a set of row outliers can be treated as column outliers of the transposed matrix. Similarly, a clustering of the columns of the transposed matrix induces a clustering of the rows of the original matrix. However, we want to find a clustering of the columns of the original matrix. To do this, one can add extra constraints on the clustering of the transposed matrix so that a clustering of the columns implicitly induces a clustering of the rows as well of the transposed matrix. This latter clustering induces a clustering of the columns of the original matrix and is our desired clustering.     

%

%\todo[inline]{To define what is \classFPT}
In particular, the theorem implies that the problem is \classFPT parameterized by $B$ and $|\Sigma |$. 
  We note that the running time of Theorem~\ref{theorem:WCO} matches the running time in  
\cite{FominGP20} obtained for the \probAtMostClust problem without outliers on binary data. 
  The interesting feature of the theorem is that the running time of the algorithm does not depend on the number of outliers $\ell$, matching the bound of the problem without outliers. Most of the clustering procedures in robust statistics, data mining and machine learning perform  well only for small number of outliers.  Our theorem  implies that if  all of the inlier points could be naturally partitioned into $k$ distinct clusters
with small cost, then such a clustering could be efficiently recovered even after arbitrarily many outliers are added. 

\vspace{-2mm}
%\subsection{Related Work} 
\subparagraph*{Related Work.} 
%Most of the previous work on feature selection are heuristics. 
%Approximation algorithms with proven guarantee for feature selection are given in
 %\cite{DBLP:conf/nips/BoutsidisMD09,BoutsidisZMD15,cohen2015dimensionality}. To the best of our knowledge, feature selection was not studied from the perspective of parameterized complexity. 
\textsc{Constrained Clustering} (without outliers) was introduced in   \cite{FominGLP020} as a tool for designing EPTAS for  \probBFact. 
\probRLRMA is a variant of robust PCA for categorical data.
\iffalse
 Let us recall that in the 
classical \emph{principal component analysis} (PCA), which  is 
one of the most popular and successful techniques used for dimension reduction  in data analysis and machine learning \cite{pearson1901liii,hotelling1933analysis,eckart1936approximation}, one seeks for the best   low-rank approximation of data matrix $M$ by solving
\begin{eqnarray*}
\text{ minimize } \|\bm{A}-\bm{B}\|^2_F \\ \text{ subject to } \rank(\bm{B}) \leq  r.
\end{eqnarray*}
Here $||\bm{A}||_F^2 = \sum_{i, j} a_{ij}^2$ is the square of the Frobenius norm of matrix $\bm{A}$.
\fi
The study of robust PCA, where one seeks for a PCA when the input data is noisy or corrupted, is 
 the large class of extensively studied  problems, see 
  the books \cite{VidalMS16,bouwmans2016handbook}. %\cite{VaswaniN18,XuCS10,bouwmans2016handbook}
There are many models of robustness in the literature, most relevant to our work
is the approach that became popular after the work of  
%Cand{\`{e}}s et al.  
 \cite{CandesLMW11}.  \iffalse  See also  \cite{WrightGRPM09,ChandrasekaranSPW11}. 
A significant body of work on the robust PCA  problem has been centered around proving that, under some feasibility assumptions on the matrices, it is possible to construct a 
 convex relaxation of the problem that can recover the matrix $\bm{B}$ efficiently.  However in most of the natural settings the robust variant becomes NP-hard \cite{GillisV15,DanHJWZ15,DSimonovFGP19}.
 \fi
The variant of robust PCA where one seeks for identifying a set of outliers, also known as PCA with outliers, were studied in 
\cite{BhaskaraK18,chen2011robust,XuCS10,DSimonovFGP19}.    
  
For the vanilla variant, \probLRMA,
\iffalse deciding whether there is a  matrix $\bm{B}$ of GF(p)-rank at most $r$ such that $\|\bm{A}-\bm{B}\|_0\leq B$,   \fi
  a number of parameterized and approximation algorithms were developed \cite{BanBBKLW19,FominGLP020,FominGP20,KumarPRW19}. 

 \probBFact  
  has  attracted much attention, especially in the data mining and knowledge discovery communities. 
In data mining, matrix decompositions are often used to produce concise representations of data. 
Since much of the real data  
 is binary or even Boolean in nature, 
Boolean low-rank approximation could provide a deeper insight into 
the semantics associated with the original matrix. There is a big body of work done on \probBFact.  
We refer to
 \cite{LuVAH12,MiettinenMGDM08,DBLP:conf/kdd/MiettinenV11,Miettinen020} for further references on this interesting problem.
 Parameterized algorithms for   \probBFact  (without outliers) were studied in \cite{FominGP20}.

There are several approximations and parameterized algorithms known for \probAtMostClust, which is the 
  vanilla (without outliers) case of \probClusteringWCO  and with $\Sigma=\{0,1\}$ \cite{OstrovskyR02,FominGLP020,BanBBKLW19,fomin_et_al:LIPIcs:2019:11576}. Most relevant to our work is the parameterized algorithm for  \probAtMostClust from  \cite{FominGP20}.  Theorem \ref{theorem:WCO}  extends the result from  \cite{FominGP20} to clustering with outliers.

\vspace{-2mm}

%\subsection{Paper Outline}
\subparagraph*{Paper Outline.}

\iffull
The remaining part of this work is organized as follows.
We briefly outline our techniques in Section~\ref{sec:techniques}.
%Then, we have some definitions in Section~\ref{sec:prelims}.
In Section \ref{sec:constrained} we describe the \classFPT algorithm for \probconstrainedcl.
The improved \classFPT algorithm for \probClusteringWCO follows in Section~\ref{sec:columnoutliers}.
Next, in Section~\ref{sec:rowoutliers}  we present the reduction from \probClusteringWRO to \probconstrainedcl.
The hardness results are deferred to Section~\ref{sec:hardness}.
Finally, in Section~\ref{sec:concl}, we conclude with some open problems.
%Due to space constraints, the improved \classFPT algorithm for \probClusteringWCO and the hardness results concerning \probClusteringWRO appear in Appendix~\ref{sec:columnoutliers} and Appendix~\ref{sec:hardness}, respectively. 
\else
In the remaining part of this extended abstract we focus on our algorithmic results.
%techniques, while the proofs of our hardness results and the reductions between problems are deferred to the attached full version of the paper.
We briefly outline our techniques in Section~\ref{sec:techniques}. Then, in Section \ref{sec:constrained} we describe our main result, the \classFPT algorithm for \probconstrainedcl. Finally, in Section~\ref{sec:concl}, we conclude with some open problems. Due to space constraints, the detailed presentation of the remaining results  
%improved \classFPT algorithm for \probClusteringWCO, the hardness results concerning \probClusteringWRO, and reductions from various problems to \probconstrainedcl, including Lemma~\ref{lem:ROsplcase}, 
appears in the attached full version of the paper.
%appear in Appendix~\ref{sec:columnoutliers} and Appendix~\ref{sec:hardness}, respectively. 
\fi

%% file: overview-organization.tex
\section{Our Techniques} \label{sec:techniques}

Both of our algorithmic results, Theorems~\ref{theorem:main} and~\ref{theorem:WCO}, have at their core the subhypergraph enumeration technique introduced by Marx \cite{Marx08}. This is fairly  natural, since our algorithms solve generalized versions of the vanilla binary clustering problem, and the only known FPT algorithm~\cite{FominGP20} for the latter problem parameterized by $B$ relies on the hypergraph enumeration as well. In fact, our algorithm for \probClusteringWCO closely follows this established approach of applying the hypergraph construction to clustering problems (\cite{fomin_et_al:LIPIcs:2019:11576}, and partly \cite{FominGP20}). However, for
the \probconstrainedcl problem the existing techniques do not work immediately. To deal with this, we generalize the previously used hypergraph construction.
In what follows, we present the key ideas of both algorithms. We begin with the simpler case of \probClusteringWCO, even though our main results are for \probClusteringWRO and \probconstrainedcl.

For the presentation of our algorithms, we recall standard hypergraph notations and the notion of a fractional cover of a hypergraph.
\iffalse
\section{Preliminaries}\label{sec:prelims}
\noindent\textbf{Hypergraphs.}
\fi
A hypergraph $G(V_G , E_G )$ consists of a set $V_G$ of vertices and a set $E_G$ of edges, where each edge is a subset of $V_G$. Consider two hypergraphs $H(V_H , E_H )$ and $G(V_G , E_G )$. We say that $H$ appears in $G$ at $V'\subseteq V_G$ as a partial hypergraph if there is a bijection $\pi$ between $V_H$ and $V'$ such that for any edge $e\in E_H$, $\pi(e)\in E_G$, where $\pi(e)=\cup_{v\in e} \pi(v)$. $H$ is said to appear in $G$ at $V'\subseteq V_G$ as a subhypergraph if there is a bijection $\pi$ between $V'$ and $V_H$ such that for any edge $e\in E_H$, there is an edge $e'\in E_G$ such that  $\pi(e)=e'\cap V'$. 

A fractional edge cover of $H$ is an assignment $\phi: E_H\rightarrow [0,1]$ such that for every vertex $v\in V_H$, the sum of the values assigned to the edges that contain $v$ is at least 1, i.e, $\sum_{e\ni v} \phi(e)\ge 1$. The fractional cover number ${\rho}^*(H)$ of $H$ is the minimum value $\sum_{e\in E} \phi(e)$ over all fractional edge covers $\phi$ of $H$. The following theorem is required for our algorithm.

\begin{theorem}\label{th:hypergraph}\cite{Marx08}
Let $H(V_H , E_H )$ be a hypergraph with fractional cover number
${\rho}^*(H)$, and let $G(V_G , E_G )$ be a hypergraph where each edge has size at most $L$. There
is an algorithm that enumerates, in time ${\vert V_H \vert}^{ O(\vert V_H\vert  )} \cdot L^{\vert V_H \vert {\rho}^*(H)+1}\cdot {\vert E_G \vert}^{{\rho}^*(H)+1}\cdot {\vert V_G \vert}^2$, 
every subset $V \subseteq V_G$ where $H$ appears in $G$ as a subhypergraph. 
\end{theorem}
%\vspace{-3mm}

\subsection{The Algorithm for \probClusteringWCO }

%first we identify \emph{initial clusters}, that is, sets of columns of $\bm{A}$ which are identical.
Given an instance ($\bm{A}$, $k$, $B$, $\ell$) of \probClusteringWCO, we note that at most $2B$ distinct columns can belong to ``nontrivial''  clusters (with at least 2 distinct columns), exactly like in the case without the outliers. So we employ a color-coding scheme \cite{AlonYZ95} to partition the columns in a way so that every column belonging to a nontrivial cluster of a fixed feasible solution is colored with its own color.
%that no two columns from different parts end up in the same cluster of the optimal solution.
Thus we reduce to multiple instances of the problem we call Restricted Clustering. In Restricted Clustering, we are given sets of columns $U_1,U_2,\ldots,U_p$ and a  parameter $B$. The goal is to select $p$ columns $\bm{b_1},\bm{b_2},\ldots,\bm{b_p}$ and a cluster center $\bm{s}$ such that $\bm{b_t}\in U_t$ for $1\le t\le p$ and $\sum_{t=1}^{p} d_H(\bm{b_t},\bm{s}) \le B$.  

Restricted Clustering is similar to the Cluster Selection problem of \cite{fomin_et_al:LIPIcs:2019:11576} and \cite{FominGP20}, and the hypergraph-based algorithm to solve it is essentially the same as in \cite{fomin_et_al:LIPIcs:2019:11576}. However, next we briefly sketch the details, as this construction serves as the base for our more general \probconstrainedcl algorithm. First, guess $\bm{b_1} \in U_1$ in the optimal solution to the instance of Restricted Clustering. If the cost of the optimal solution is at most $B$, then $d_H(\bm{b_1}, \bm{s})$ is at most $B$ as well. If we know the set of at most $B$ positions $P$ where $\bm{b_1}$ and $\bm{s}$ differ, we can easily identify $\bm{s}$ by trying all possible $|\Sigma|^B$ options at these positions. For each option, we can find the closest $\bm{b_i}$ from each $U_i$ and check whether the total cost is at most $B$.
%For that, we identify the input instance with a certain hypergraph. 

To show that we can enumerate all potential sets of deviating positions in FPT time, we identify the instance with the following hypergraph $H(V, E)$. The vertices $V$ are the positions $\{1, \ldots, m\}$. With each column $\bm{x}$ in $\bigcup_{i = 1}^p U_i$, we identify a hyperedge containing exactly the positions where $\bm{x}$ is different from $\bm{b_1}$. Now the optimal set of positions $P$ and the optimal columns $\{\bm{b_i}\}_{i=1}^p$ induce a subhypergraph $H_0(V_0, E_0)$ of $H$ such that $|V_0| \le B$ and the fractional cover number  of $H_0$ is at most two.
The latter holds simply because wherever $\bm{s}$ is different from $\bm{x}$, at least half of the chosen columns must also be different from $\bm{x}$, otherwise modifying $\bm{s}$ to match $\bm{x}$ decreases the cost.
If we enumerate all possible subhypergraphs $H_0$ and all possible locations in $H$ where they occur, we can surely find the optimal set of locations $P$. Since $|V_0| \le B$, enumerating all choices for $H_0$ is clearly in FPT time. For a fixed $H_0$, finding all occurrences in $H$ is in FPT time by Theorem~\ref{th:hypergraph}.
  Note that applying Theorem~\ref{th:hypergraph} results in FPT time only when the fractional cover number of $H_0$ is bounded by a constant. Also, by a sampling argument %(Lemma~\ref{lem:hg-wid-2nd-property}), 
  one can show that it suffices to consider only those $H_0$ with $O(\log B)$ hyperedges.
It follows that the number of distinct hypergraphs that we have to consider for enumeration is bounded by only $2^{O(B\log B)}$. Thus it is possible to bound the dependence on $B$ in the running time by $2^{O(B\log B)}$.

\iffull
A specific issue we need to deal with in the \probClusteringWCO problem is the following. If we encounter a number of identical columns, we need to treat them as an indivisible set (\emph{initial cluster}), as only a cluster with sufficiently many different columns can have large cost. However, a potential flaw of this approach is that in an optimal solution, identical columns can belong to different clusters and to the set of outliers. To handle this, we prove that the optimal solution can always be ``normalized''. Namely, we can rearrange clusters of the optimal solution in a way that at most one  initial cluster is split between a solution cluster and the set of outliers, and every other initial cluster is completely contained in either one of the resulting clusters or the set of outliers. Then, guessing the one initial cluster being split solves the issue.
After the color-coding and solving Restricted Clustering, each remaining initial cluster is either a trivial cluster in the solution, or belongs to the outliers. This can be decided greedily.
\else
\fi

\subsection{The Algorithm for \probClusteringWRO}% and \probconstrainedcl}

For \probClusteringWRO, the above-mentioned approach is not applicable, for several reasons. Most crucially, it does not seem that one can partition the problem into $k$ independent instances of a simpler single-center selection problem, in a way that we reduce \probClusteringWCO to $k$ instances of Restricted Clustering (for a fixed coloring). Intuitively, the possibility to remove a subset of features does not allow such a partition as all the clusters depend simultaneously on the choice of those features.
Moreover, our hardness result shows that for \probClusteringWRO the running time cannot match with \probClusteringWCO, as no $n^{o(k)}$ time algorithm is possible for constant $B$, assuming ETH.

%Intuitively, the hardness of \probClusteringWRO lies in the fact that we have to select the outlier rows and to group columns into clusters at the same time. So first, to make the problem more accessible, we give a novel reduction to \probconstrainedcl, where both partitioning into clusters and removing the outliers are with respect to the columns, at the cost of having arbitrary row-wise constraints on the cluster centers. 
%Recall also that solving \probconstrainedcl immediately gives  algorithms for a variety of other problems. 
%The general idea of the reduction from \probClusteringWRO to \probconstrainedcl is as follows. For any matrix, a set of row outliers can be treated as column outliers of the transposed matrix. Similarly, a clustering of the columns of the transposed matrix induces a clustering of the rows of the original matrix. However, we want to find a clustering of the columns of the original matrix. To do this, one can add extra constraints on the cluster centers so that a clustering of the columns of the transposed matrix implicitly induces also a clustering of its rows. Essentially, grouping rows into $k$ clusters corresponds to grouping columns into $|\Sigma|^k$ clusters such that the centers of these column clusters can be ``read off'' from the $k$ row centers.
%Simply by transposing, the clustering of the columns of the transposed matrix induces a clustering of the rows of the original matrix, which is our desired clustering. This is the crux of our reduction.     

By Lemma~\ref{lem:ROsplcase}, for solving \probClusteringWRO, it suffices to solve
 \probconstrainedcl. For the same reasons as with \probClusteringWRO, the approach used for \probClusteringWCO fails, as the constraints on the centers do not allow to form clusters independently.
Instead, we generalize the hypergraph construction used for Restricted Clustering to handle the choice of all $k$ centers simultaneously, as opposed to just one center at a time. This is the most technical part of the paper.
%, and we will now sketch the generalized construction briefly. 
 The main idea is to base the hypergraph on $k$-tuples of columns instead of just single columns. In the next section,
 we formalize this intuitive discussion, presenting the proof in full detail.

%% file: algorithm-constrained-clustering.tex
\section{The Algorithm for \probconstrainedcl}\label{sec:constrained}
In this section we prove Theorem~\ref{theorem:main} by giving an algorithm for  
 \probconstrainedcl that runs in $(kB)^{O(kB)}n^{O(k)}m^2{\vert \Sigma\vert}^{kB}$ time. 
 %As mentioned before, we do not directly solve \probconstrainedcl. Instead, for the sake of convenience, we consider a different, but very similar problem. In the following, we formally define this new problem.  %In this section, we describe our algorithm for group clustering. 
First, we prove a structural lemma that will be useful for analysis of the algorithm. 
%Note that it is sufficient to return a set of $r$ centers for a \yesinstance, as a clustering of the minimum cost can be computed using the centers in a greedy manner. 
%and assume WLOG that $\vert I_1\vert \ge \vert I_2\vert \ge \ldots \ge \vert I_r\vert (\ge 1)$. As there are only $n^{O(r)}$ different choices of the sizes of the clusters, we can assume that the size of $I_i$ is known for all $1\le i\le r$. Fix a constant $\alpha > 1$. Later we will set $\alpha$ to 2. First, we compute a \say{balanced} partition of $\{1,\ldots,r\}$. 

%\vspace{-2mm}

\subsection{Structural Lemma}

Here we show that an optimal set of centers corresponds to 
%In this section, we prove a structural lemma that shows the existence of
a \say{good} subhypergraph in a certain hypergraph. Consider a feasible clustering $\mathcal{I}=\{I_1,I_2,\ldots,I_k\}$ having the minimum cost.
Let $\{\bm{c}_1,\bm{c}_2,\ldots,\bm{c}_k\}$ be a fixed set of centers corresponding to $\mathcal{I}$. Also, let $T$ be the set of all tuples of the form $(\bm{a}^{i_1},\bm{a}^{i_{2}},$ $\ldots,\bm{a}^{i_{k}})$ such that $i_j \in I_j$ for all $j$. Note that we do not actually need to know the set $T$ --- we just introduce the notation for the sake of analysis. For a $k$-tuple $x=(\bm{x}_1,\ldots,\bm{x}_k)$, we denote the tuple $(\bm{x}_1[j],\bm{x}_2[j],\ldots,\bm{x}_k[j])$ by $x[j]$. Two $k$-tuples $x$ and $y$ are said to differ from each other at location $j$ if $x[j]\ne y[j]$.
%For a $k$-tuple $z=(z_1,\ldots,z_k)$ and indices $1\le j^1 < j^2\le k$, we denote the tuple $(z_{j^1},z_{j^1+1},\ldots,z_{j^2})$ by $z[j^1\ldots j^2]$. 

Let $C$ be the $k$-tuple such that $C=(\bm{c}_1,\bm{c}_2,\ldots,\bm{c}_k)$. Suppose $x=(\bm{x}_1,\ldots,\bm{x}_{k})$ is such that 
%for $1\le j\le p$ (for $i > 1 $), $\bm{x}_j=\bm{c}_j$ and for all $p+1\le j\le r$, the column $\bm{x}_j$ is in cluster $I_j$
there are at most $B$ positions $h$ where $x[h]\ne C[h]$, and for each $1\le j\le m$, $x[j]\in R_j$. 
%Moreover, $\bm{x}_j=\bm{c}_j$ for $1\le j\le p$ (if $p\ge 1$). 
Consider the hypergraph $H$ defined in the following way with respect to $x$. The labels of the vertices of $H$ are in $\{1,2,\ldots,m\}$. For each $k$-tuple $y=(\bm{y}_1,\ldots,\bm{y}_k)$ of $T$, we add an edge $S\subseteq \{1,2,\ldots,m\}$ such that $h\in S$ if ${x}[h]\ne {y}[h]$.
%, i.e, $(\bm{x}_1,\ldots,\bm{x}_{p+q})$ and $(\bm{y}_1,\ldots,\bm{y}_{p+q})$ differs from each other at location $h$. Finally, let $f(\alpha,r)=\frac{\alpha(r(\alpha+1)^{r}+1)}{\alpha-1}$. 

In the following lemma, we show that the hypergraph $H$ has a \say{good} subhypergraph. 

\begin{lemma}\label{lem:existenceofhypergraph}
(Structural Lemma) Suppose the input is a \yesinstance. Consider a $k$-tuple $x=(\bm{x}_1,\ldots,\bm{x}_{k})$ such that there are at most $B$ positions $h$ where $x[h]\ne C[h]$ and for each $1\le j\le m$, $x[j]\in R_j$. 
%Moreover, $\bm{x}_j=\bm{c}_j$ for $1\le j\le p$. 
Also, consider the hypergraph $H$ defined in the above with respect to $x$. There exists a subhypergraph $H^*(V^*,E^*)$ of $H$ with the following properties.
\begin{enumerate}
    \item $\vert V^*\vert \le B$.
    \item $\vert E^*\vert \le 200\ln B$.
    \item The indices in $V^*$ are the exact positions $h$ such that $x[h]\ne C[h]$. 
    \item The fractional cover number of $H^*$ is at most $4$.
\end{enumerate}
\end{lemma}

To prove the above lemma, first, we show the existence of a subhypergraph that satisfies all the properties except the second one. Let $P$ be the set of positions $h$ such that $x[h]\ne C[h]$. Let $H_0(V_0,E_0)$ be the subhypergraph of $H$ induced by $P$.
%We show that $H_0$ satisfies the first, third and fourth properties mentioned in Lemma \ref{lem:existenceofhypergraph}. 
By definition of $x$, $P$ contains at most $B$ indices. Thus, the first property follows immediately. The third property also follows, as $V_0 = P$, is exactly the set of positions  $h\in \{1, 2, \ldots, m\}$, where ${x}[h]\ne C[h]$.
Next, we show that the fourth property holds for $H_0$. In fact, we show a stronger bound.  

\begin{lemma}
The fractional cover number of $H_0$ is at most $2$.
\end{lemma}

\begin{proof}
Note that the total number of edges of $H_0$ is $\tau=\vert T\vert$. We claim that each vertex of $H_0$ is contained in at least $\tau/2$  edges. 

Consider any vertex $h$ of $H_0$. Suppose there is a $1\le j\le k$, such that for at least $\lceil |I_j|/2 \rceil$ columns in $I_j$ the value at location $h$ is not $\bm{x}_j[h]$. Note that each such column contributes to $\Pi_{t^1=1}^{j-1} |I_{t^1}|\cdot \Pi_{t^2=j+1}^{k} |I_{t^2}|=\tau/|I_j|$ tuples $(\bm{y}_1,\ldots,\bm{y}_k)$ of $T$ such that $\bm{y}_j[h]\ne \bm{x}_j[h]$. Thus, the edge corresponding to each such tuple contains $h$. It follows that, at least  $\lceil |I_j|/2 \rceil\cdot (\tau/|I_j|)\ge \tau/2$ edges in $E_0$ contain $h$. 

In the other case, for all $1\le j\le k$ and less than $\lceil |I_j|/2 \rceil$ columns in $I_j$, the value at location $h$ is not $\bm{x}_j[h]$. We prove that this case does not occur. Note that there is a $k$-tuple $z$ in $R_h$ such that $z=x[h]$. Consider replacing $C[h]$ by $z$ at position $h$ of $C$. Next, we analyze the change in cost of the clustering $\mathcal{I}$. Note that the cost corresponding to positions other than $h$ remains same. For a $1\le j\le k$, if previously $\bm{c}_j[h]= \bm{x}_j[h]$, the cost remains same. Otherwise, $\bm{c}_j[h]\ne \bm{x}_j[h]$. Note that for more than $\lceil |I_j|/2 \rceil$ columns in $I_j$, the value at location $h$ is $\bm{x}_j[h]$. Thus, by replacing $\bm{c}_j[h]$ by $\bm{x}_j[h]$, the cost decrement corresponding to the index $j$ and location $h$ is at least $1$. As $x[h]\ne C[h]$, there is an index $j$ such that $\bm{c}_j[h]\ne \bm{x}_j[h]$. It follows that the overall cost decrement is at least $1$, which contradicts the optimality of the previously chosen centers. Hence, this case cannot occur. This completes the proof of the lemma.  
\end{proof}

So far we have proved the existence of a subhypergraph that satisfies all the properties except the second. Next, we show the existence of a subhypergraph that satisfies all the properties. The following lemma completes the proof of Lemma \ref{lem:existenceofhypergraph}. %, the proof of which appears in Appendix \ref{ap:hg-wid-2nd-property}.
\iffull
%Its proof follows a standard sampling argument.
\else
Its proof follows a standard sampling argument, and is presented in the full version. %Appendix \ref{ap:hg-wid-2nd-property}.
\fi

\begin{lemma}\label{lem:hg-wid-2nd-property}
Let $B \ge 2$. Consider the subhypergraph $H_0$ that satisfies all the properties of Lemma \ref{lem:existenceofhypergraph} except (2). It is possible to select at most $200\ln B$ edges from $H_0$ such that the subhypergraph $H_0^*$ obtained by removing all the other edges from $H_0$ satisfies all the properties of Lemma \ref{lem:existenceofhypergraph}. 
\end{lemma}

%\section{Proof of Lemma \ref{lem:hg-wid-2nd-property}}\label{ap:hg-wid-2nd-property}
\iffull
For the proof of the lemma, we need the following form of Chernoff bound. 

\begin{theorem}\label{th:chernoff}\cite{AngluinV77}
Let $X_1,\ldots,X_n$ be independent 0-1 random variables with Pr$[X_i=1]=p_i$. Let $X=\sum_{i=1}^n X_i$ and $\mu=E[X]$. Then \text{ for } $0< \beta \le 1$,

\begin{align*}
& Pr[X\le (1-\beta)\mu] \le exp(-\beta^2\mu/2)\\
& Pr[X\ge (1+\beta)\mu] \le exp(-\beta^2\mu/3)
%,\\& Pr[X\ge (1+\beta)\mu] \le exp(-\beta^2\mu/(2+\beta)) & \text{ for } \beta > 1
\end{align*}

\end{theorem}

Now, the proof of Lemma \ref{lem:hg-wid-2nd-property} is as follows. 

\begin{proof}
Recall that $\tau$ is the number of edges in $H_0$. Construct a hypergraph $H_0'$ by selecting  each edge of $H_0$ independently at random with probability $\frac{150\ln B}{\tau}$. The expected number of edges selected is $150\ln B$. Also, by Theorem \ref{th:chernoff} with $\beta =1/3$, the probability that at least $200\ln B$ edges are selected is at most $exp(-150\ln B/27)<1/B^2$. Each vertex of $H_0$ is contained in at least $\tau/2$  edges. Thus, the expected number of edges that cover a vertex of $H_0'$ is at least $75\ln B$. Moreover, by Theorem \ref{th:chernoff} with $\beta =1/3$, the probability that a given vertex of $H_0^*$ is covered by at most $50\ln B$ edges is at most $exp(-75\ln B/18) < 1/B^3$. Therefore, with probability at least $1-1/B^2-B\cdot 1/B^3$, $H_0'$ contains at most $200\ln B$ edges and each vertex of $H_0'$ is covered by at least $50\ln B$ edges. Thus, the fractional cover number of $H_0'$ is at most $4$. Hence, the second and the fourth property are satisfied. As the vertex set does not get changed, the first and the third property also hold. It follows that there exists a subhypergraph $H_0^*$ that contains at most $200\ln B$ edges from $H_0$ and satisfies all the properties of Lemma \ref{lem:existenceofhypergraph}.
\end{proof}
\fi

%\vspace{-10pt}
\subsection{The Algorithm for Constrained Clustering}
In this section, we describe our algorithm. %Note that the goal of our algorithm is to retrieve a feasible clustering if there is one. 
The algorithm outputs a feasible clustering of minimum cost if there is a feasible clustering of the given instance. Otherwise, the algorithm returns \say{NO}.

%\vspace{-10pt}
\subparagraph*{The Algorithm.}

First, we consider all $k$-tuples $x=(\bm{x}_1,\ldots,\bm{x}_k)$ such that  $\bm{x}_j$ is a column of $\bm{A}$ for $1\le j \le k$, and apply the following refinement process on each of them. 
Here, a $k$-tuple $x$ of columns of $\bm{A}$ is said to differ from $\R$ at a position $j$ for $1\le j\le m$ if $x[j]\notin R_j$. 

%\vspace{-6pt}
\begin{itemize}
    %\item For each $r$-tuple $x=(\bm{x}_1,\ldots,\bm{x}_r)$ of columns of $\bm{A}$ apply the following procedure.
    \item Let $P \subseteq \{1,2,\ldots,m\}$ be the set of positions where $x$ differs from $\R$.
    \item If $\vert P\vert > B$, probe the next $k$-tuple $x$. 
    \item For each position $h \in P$, replace $x[h]$ by any element of $R_h$. 
   % \item Output all the refined tuples. 
\end{itemize}
%\vspace{-4pt}

Next, for each refined $k$-tuple $x=(\bm{x}_1,\ldots,\bm{x}_k)$, we construct a hypergraph $G$ whose vertices are in $\{1,2,\ldots,m\}$. For each $k$-tuple $y=(\bm{y}_1,\ldots,\bm{y}_k)$ of columns of $\bm{A}$, we add an edge $S\subseteq \{1,2,\ldots,m\}$ such that $h\in S$ if ${x}[h]\ne {y}[h]$. For all hypergraphs $H_0^*$ having at most $B$ vertices and at most $200\ln B$ edges, we check if each vertex of $H_0^*$ is contained in at least $1/4$ fraction of the edges. If that is the case, we use the algorithm of Theorem \ref{th:hypergraph} to find every place $P'$ where $H_0^*$ appears in $G$ as subhypergraph. For each such set $P'$, we perform all possible $B'\le B\cdot k$ edits of the tuple ${x}$ at the locations in $P'$. In particular, for each $B'$, the editing is done in the following way. For each possible $B'$ entries $(a_1,\ldots,a_{B'})$ in $(\bm{x}_{1},\ldots,\bm{x}_{k})$ at the locations in $P'$ and each set of $B'$ symbols $(s_1,\ldots,s_{B'})$ from $\Sigma$, we put $s_j$ at the entry $a_j$ for all $j$. After each such edit, we retrieve the edited tuple $(\bm{x}_{1},\ldots,\bm{x}_{k})$ and perform a sanity check on this tuple to ensure that it is a valid $k$-tuple center. In particular, for each index $1\le h\le m$, if there is a $z\in R_h$ such that $z=x[h]$, we tag $x$ as a valid tuple. Lastly, we output all such valid $k$-tuples as candidate centers if the corresponding cost of clustering is at most $B$. If no $k$-tuple is output as a candidate center, we return \say{NO}.   

Note that, given a $k$-tuple candidate center $z=(\bm{z}_1,\ldots,\bm{z}_k)$, one can compute a minimum cost clustering in the following greedy way, which implies that  we can correctly compute the cost of clustering in the above. At each step $i$, we assign a new column of $\bm{A}$ to a center. In particular, we add a column $\bm{a}^j$ of $\bm{A}$ to a cluster $I'_t$ such that $\bm{a}^j$ incurs the minimum cost over all unassigned columns if it is assigned to a center, i.e, it minimizes the quantity $\min_{t'=1}^k  d(\bm{a}^j,\bm{z}_{t'})$, and $\bm{z}_t$ is a corresponding center nearest to $\bm{a}^j$. As we are allowed to exclude $\ell$ outliers, we assign $n-\ell$ columns. The clustering $\{I'_1,\ldots,I'_k\}$ obtained at the end of this process is the output. This finishes the description of our  algorithm.     

\vspace{-2mm}

\subsection{Analysis}

Again consider the feasible clustering with partition $\mathcal{I}=\{I_1,I_2,\ldots,I_k\}$ and the corresponding center tuple $C=(\bm{c}_1,\bm{c}_2,\ldots,\bm{c}_k)$ having the minimum cost. 
First, we have the following observation. 

\begin{observation}\label{obs:refine}
Suppose for a $k$-tuple $x$, $x$ differs from $C$ at $B_1$ positions. Then, after refinement, there is at most $B_1$ positions $h$ such that $x[h]\ne C[h]$. Moreover, after refinement, $d_H(x,C)\le B_1\cdot k$. 
\end{observation}

The first part is true for the following reason. If $x$ was different from $\R$ at a position $h$, then $x[h]\ne C[h]$ as well. Thus, refinement is applied for a position $h$ where $x[h]$ already differs from $C[h]$. Hence, refinement does not affect a position $h$ where $x[h]=C[h]$. The moreover part follows trivially from the first part as $x$ is a $k$-tuple. Now, it is sufficient to prove the following lemma.

\begin{lemma}
Suppose there is a feasible clustering with partition $\mathcal{I}=\{I_1,I_2,\ldots,I_k\}$ as defined above. The above algorithm successfully outputs the $k$-tuple $(\bm{c}_1,\bm{c}_2,\ldots,\bm{c}_k)$. 
\end{lemma}

\begin{proof}
Consider a $k$-tuple $x=(\bm{x}_1,\ldots,\bm{x}_k)$ such that the column $\bm{x}_j$ is in cluster $I_j$. As the algorithm considers all possible $k$-tuples of columns in $\bm{A}$, it must consider $x$. 
%Note that there can be at most $B$ positions $h$ such that $x[h]\ne C[h]$. 
By Observation \ref{obs:refine}, after refinement, there are at most $B$ positions $h$ where $x[h]\ne C[h]$. Also, for each $1\le j\le m$, $x[j]\in R_j$. Let $G$ be the hypergraph constructed by the algorithm corresponding to this refined $x$. Note that the hypergraph $H$ defined in Lemma \ref{lem:existenceofhypergraph} is a partial subhypergraph of $G$. Thus, the subhypergraph $H^*$ of $H$ is also a subhypergraph of $G$. As we enumerate all hypergraphs having at most $B$ vertices, at most $200\ln B$ edges and at most $4$ fractional covering number, $H^*$ must be considered by the algorithm. Let $P'$ be the place in $G$ where $H^*$ appears. By the third property of Lemma \ref{lem:existenceofhypergraph}, the locations in $P'$ are the exact positions $h$ such that ${x}[h]\ne C[h]$. It follows that an edit corresponding to $P'$ generates the tuple $C=(\bm{c}_1,\ldots,\bm{c}_{k})$, as $d_H(x,C)\le B\cdot k$. It is easy to see that $C$ also passes the sanity check. Hence, $C$ must be an output of the algorithm. 
\end{proof}

Given the tuple center $C=(\bm{c}_1,\ldots,\bm{c}_{k})$, we use the greedy assignment scheme (described in the algorithm) to find the underlying clustering. Note that given any $k$-tuple candidate center $z=(\bm{z}_1,\ldots,\bm{z}_k)$, this greedy scheme computes a minimum cost clustering with ${\bm{z}_1,\ldots,\bm{z}_k}$ being the cluster centers. Thus, the cost of the clustering computed by the algorithm is at most $B$. Hence, the algorithm successfully outputs $C$ as a candidate center. We summarize our findings in the following lemma. 

\begin{lemma}
Suppose the input instance is a \noinstance, then the  algorithm successfully outputs \say{NO}. If the input instance is a \yesinstance, the algorithm correctly computes a feasible clustering. 
\end{lemma}

%\vspace{-5mm}
\subsection{Time Complexity}
Next, we discuss the time complexity of our algorithm. The number of choices of $x$ is $n^{O(k)}$. For each choice of $x$, the hypergraph $G$ can be constructed in $n^{O(k)}$ time. The number of distinct hypergraphs $H_0^*$ with at most $B$ vertices and at most $200\ln B$ edges is $2^{B\cdot 200\ln B}=B^{O(B)}$, since there are $2^B$ possibilities for each edge. Now we analyze the time needed for locating a particular $H_0^*$ in $G$. For any tuple $y\in T$, $d_H(y,C) \le B$. By triangle inequality, $d_H(x,y)\le 2B$. Thus, the size of any edge in $H$ is at most $2B$, and we can remove any edge of $G$ of size more than $2B$. From Theorem \ref{th:hypergraph}, it follows that every occurrence of $H_0^*$ in $G$ can be found in $B^{O(B)}\cdot {(2B)}^{ {4B+1}}\cdot {n}^{4k+k}\cdot { m }^2=B^{O(B)}\cdot n^{O(k)}m^2$ time. If $H_0^*$ appears at some place in $G$, it would take $O((kB\vert \Sigma\vert)^{kB})$ time to edit $x$. Hence, in total the algorithm takes $(kB)^{O(kB)}{\vert \Sigma\vert}^{B}\cdot n^{O(k)}m^2$ time. By the above discussion, we have Theorem~\ref{theorem:main}.
%, which we restate next. 
%Thus intuitively, \probClusteringWRO can be treated as a special case of the \probconstrainedcl. Next, we formally prove this.  

\iffull
\thmmain*
\fi
%\begin{theorem}
%\probconstrainedcl can be solved in $(rk)^{O(rk)}n^{O(r)}m^2$ ${\vert \Sigma\vert}^{rk}$ time. 
%\end{theorem}

%% file: algorithm-column-outliers.tex
 %!TEX root = column-outliers.tex
 \newcommand{\wrt}{with respect to\xspace}
\section{FPT Algorithm for \probClusteringWCO Parameterized by $B$}\label{sec:columnoutliers}In this section we give a proof of Theorem~\ref{theorem:WCO}. Let us recall that the theorem states that 
\emph{\probClusteringWCO is solvable  in time $2^{O(B\log B)}{\vert\Sigma\vert}^{B}(nm)^{O(1)}$.}

\medskip
 Let $(\bm{A},k,B,\ell)$ be an input of \probClusteringWCO, where $\bm{A}$ is an 
   $m\times n$ matrix,  
    $k$ the number of clusters,  $B$ the cost of the solution and $\ell$ the number of outliers.

Let $\beta$ be the number of pairwise distinct columns in $\bm{A}$. 
We start with partitioning the columns of $\bm{A}$ into sets of identical columns
  $\mathcal{J}=\{J_1,J_2,\ldots,J_{\beta}\}$. That is, all columns in each $J_i$ are equal.  We refer to such a set $J_i$ as an  
   \emph{initial cluster}.
%  of the $n$ columns such that each initial cluster $J_i$ contains identical columns. In case of \probClustering, one can show that, if there is a feasible clustering, then there is one with the property that for any $i$, all the columns in $J_i$ belong to same cluster. In other words, there is a feasible clustering where no initial cluster gets split. %Indeed, the algorithm in \cite{abs-1803-06102} relies on this fact. 
%Note that this might not be true in our case. 

We show that for any \yesinstance, there is a feasible solution  such that the columns of at most one 
initial cluster are ``split''  by the solution. More formally, consider a clustering $\mathcal{C}=\{I_1,I_2,\ldots,I_k\}$  of $\{1,2,\ldots,n\} \setminus O$. Then each of the initial clusters $J_i$ is exactly one of the following types \emph{\wrt $\mathcal{C}$}. 
\begin{itemize}
\item[(i)] $\{J_i\}$ such that there is $t$ with $J_i \subseteq I_t$, 
 \item[(ii)] $\{J_i\}$ such that $J_i$ is not of type (i), and $J_i \subseteq \cup_{j=1}^k I_j$, 
 \item[(iii)] $\{J_i\}$ such that $J_i \subseteq O$,
 \item[(iv)] $\{J_i\}$ such that $J_i \not\subseteq \cup_{j=1}^k I_j$,
  $J_i \not\subseteq O$, and there is $t$ with $J_i \subseteq I_t\cup O$, and 
  \item[(v)] $\{J_i\}$ such that $J_i \not\subseteq \cup_{j=1}^k I_j$, $J_i \not\subseteq O$, and there is no $t$ with $J_i \subseteq I_t\cup O$. 
 \end{itemize}
 %Note that each of the  initial clusters must be one of the above   five types. Next, we prove the following lemma.

\begin{lemma}\label{lem:atmost1split} 
Let $(\bm{A},k,B,\ell)$ be a \yesinstance of  \probClusteringWCO. Then there is a set $O\subseteq \{1,2,\ldots,n\} $ of size $\ell$ and $k$-clustering 
$\mathcal{C}$  of $\{1,2,\ldots,n\} \setminus O$ of cost at most $B$ such that 
%Suppose there is a feasible clustering for the given instance. Then, there is a feasible clustering $\mathcal{C}$ such that any  
 every
initial cluster of $\bm{A}$  \wrt $\mathcal{C}$ is of type (i), (iii) or (iv). Moreover,  there is at most one initial cluster of type (iv).  
\end{lemma}

\begin{proof}
Because $(\bm{A},k,B,\ell)$ be a \yesinstance, there is a set of outliers $O$ and a $k$-clustering of the remaining vectors of cost at most $B$. Among all such clusterings, we select a clustering  $\mathcal{C}=\{I_1,I_2,\ldots,I_k\}$    of $\{1,2,\ldots,n\} \setminus O$ of cost at most $B$ such that \wrt $\mathcal{C}$
\begin{equation}\label{eq:initcl1} \text{the number of initial clusters of type (ii) is minimum}, 
\end{equation}
subject to \eqref{eq:initcl1}
\begin{equation}\label{eq:initcl2} \text{the number of initial clusters of type (v) is minimum}, 
\end{equation}
and 
subject to \eqref{eq:initcl1} and  \eqref{eq:initcl2}, 
\begin{equation}\label{eq:initcl3} \text{the number of initial clusters of type (iv) is minimum}.
\end{equation}
We claim that $\mathcal{C}$ is the required clustering. Targeting towards a contradiction, assume first 
 that  there is an initial cluster of $\mathcal{J}$ of  type (ii) \wrt $\mathcal{C}$. Without loss of generality, let this cluster be $J_1$. 
Then  $J_1$ is not fully contained in any cluster  of  $\mathcal{C}$, but it is fully contained in the union of all clusters. We convert  $J_1$ into a type (i) initial cluster by the following modification of  $\mathcal{C}$. 
Let $\bm{c}_i$ be the center of  cluster $I_i$ for  $i\in\{1,\dots,k\}$. Also, let $\bf{a}$ be the unique column corresponding to $J_i$ (the whole initial cluster consists of columns equal to $\bf{a}$). Let $\bm{c}_{q}$ be a center that minimizes the cost
 $d_H({\bf{a}},{\bm{c}}_i)$ over all centers ${\bm{c}}_i$. We construct a new clustering $\mathcal{C}'$  by  reassigning the columns of $J_1$ to the cluster $I_{q}$. Thus in the new clustering $J_1$ becomes of type (i). The assignment of all other initial clusters remains the same. 
Since $d_H({\bf{a}},{\bm{c}}_q)\leq  d_H({\bf{a}},{{\bm{c}}_1})$, we have that the cost of the new clustering $\mathcal{C}'$ is at most $B$, thus contradicting assumption  \eqref{eq:initcl1}. Hence, there are no initial clusters of $\mathcal{J}$ of  type (ii) \wrt $\mathcal{C}$.

 By applying the same arguments  to the non-outlier part of a type (v) initial cluster \wrt $\mathcal{C}$.  The only difference that now type (v) initial cluster becomes type (iv)  thus contradicting assumption  \eqref{eq:initcl2}.
 
 Finally, we already know that there are no  types (ii) and (v) initial clusters \wrt  $\mathcal{C}$. Assume that there are at least two, say $J_1$ and $J_2$, initial clusters of type (iv) \wrt $\mathcal{C}$.
Then, there are clusters $C_{i_1}$ and $C_{i_2}$ in $\mathcal{C}$ such that $J_1 \subseteq C_{i_1}\cup O$ and $J_2 \subseteq C_{i_2}\cup O$.  (We do not exclude the possibility of $i_1=i_2$ here.) Let ${\bf{c}}_{i_1}$  and  ${\bf{c}}_{i_2}$ be the centers of  $C_{i_1}$ and $C_{i_2}$ respectively.
Let also  ${\bf{a}}_1$ and   $\bf{a}_2$  be the unique columns corresponding to $J_1$ and  $J_2$. 
%:
Without loss of generality, we assume that  $d_H({\bf{a}}_1,{\bf{c}}_{i_1})) \le d_H({\bf{a}}_2,{\bf{c}}_{i_2}))$. There are two cases: (1) $\vert J_2\cap C_{i_2}\vert \le \vert J_1\cap O\vert$, and (2) $\vert J_2\cap C_{i_2}\vert > \vert J_1\cap O\vert$. In the first case, we can assign additional $\vert J_2\cap C_{i_2}\vert$ columns of $J_i$ to $C_{i_1}$, and move  $\vert J_t\cap C_{i_2}\vert$ columns of $J_2$ from $C_{i_2}$ to $O$. The cost increase is $\vert J_2\cap C_{i_2}\vert \cdot (d_H({\bf{a}}_1,{\bf{c}}_{i_1}) - d_H({\bf{a}}_2,{\bf{c}}_{i_2}))\le 0$. Also, type of $J_2$ is now changed to (iii) \wrt the new clustering. Thus, the number of type (iv) initial clusters has decreased by at least 1 and no type (ii) or (iv) clusters were created. 
For  the second case, assign the $\vert J_1\cap O\vert$ columns of $J_1$ from $O$ to $C_{i_1}$, and move $\vert J_1\cap O\vert$ more columns of $J_2$ from $C_{i_2}$ to $O$. The cost increase is $\vert J_1\cap O\vert \cdot (d_H({\bf{a}}_1,{\bf{c}}_{i_1}) - d_H({\bf{a}}_2,{\bf{c}}_{i_2}))\le 0$. Also, type of $J_1$ is now changed to (i) \wrt the new clustering. Thus, in this case also, the number of type (iv) initial clusters has decreased by 1. 
In both cases we constructed  clusterings contradicting  \eqref{eq:initcl3}. Hence there is at most one initial cluster of type (iv) \wrt 
 $\mathcal{C}$. 
\end{proof}

Now, we describe our algorithm. The algorithm tries to find a feasible clustering (if any) that satisfies the property of Lemma \ref{lem:atmost1split}. As there is at most one initial cluster that can get split, we can guess it in advance. Suppose $J_t$ be this initial cluster. 
%WLOG, we can assume that the unique column in $J_t$ is not a center. Otherwise, there is no need to split $J_t$. 
We can also guess the number of columns $\ell'$ of $J_t$ that would be part of the outliers set. We construct a new instance of the problem which contains all the columns in $J_i$ for $1\le i\le \beta$ and $i\ne t$, and exactly $|J_t|-\ell'$ columns of $J_t$. 
%the cost $k'$ incurred by the remaining columns of $J_t$ that would be part of the clustering. Hence, for the time being, we do not need to consider $J_t$. We can find a clustering of the remaining initial clusters with at most $\ell-\ell'$ outliers such that no initial cluster gets split and the cost is at most $k-k'$. Once we find the clustering we can assign the appropriate number of columns of $J_t$ to a cluster that minimizes the total cost and that would give us the complete solution. 
Let $\bm{A}'$ be the matrix $\bm{A}$ containing these columns. Then, our new instance is $(\bm{A}',k,B,\ell-\ell')$. The advantage of working with this instance is that in the desired clustering no initial cluster gets split. Next, we show how to find such a clustering for this instance. 

%The rest of the algorithm follows the footsteps of the one in \cite{abs-1803-06102}. However, in our case we need to figure out the outliers and matrix entries are not just binary. 
%For completeness, here we describe the algorithm. 
For simplicity of notations, let us denote the input instance by $(\bm{A},k,B,\ell)$. As no initial cluster in the desired clustering gets split, an initial cluster can either form its own cluster, or becomes part of the outliers, or gets merged with other initial clusters to form a new cluster. We refer to the last type of clusters as composite clusters. Now, we have the following simple observation.

\begin{observation}\label{obs:noofcompclusters}
In a feasible clustering, the total number of initial clusters that can be part of the composite clusters is at most $2B$. The total number of composite clusters is at most $B$.
\end{observation}

Next, we use the color coding technique of \cite{AlonYZ95} to randomly color the initial clusters by one of the colors in $\{1,2,\ldots,2B\}$. The idea is that with sufficient probability the initial clusters that form the composite clusters get colored by pairwise distinct colors. Thus, assuming this event occurs, it is sufficient to select only one initial cluster from each color class to find out such initial clusters. However, we do not know how to combine those initial clusters. Nevertheless, we can guess such combination by trying all possible partitions of $\{1,2,\ldots,2B\}$ that contains at most $B$ parts. For each part, we can also guess the cost of the corresponding composite cluster. By Observation \ref{obs:noofcompclusters}, the selected initial clusters corresponding to each part form a composite cluster. Each initial cluster which is not selected would either form its own cluster or become part of the outliers. However, an initial cluster that forms its own cluster does not incur any cost. Hence, any subsets of the appropriate number of non-selected initial clusters can form their own clusters without incurring any extra cost. However, we need to ensure that the number of outliers is at most $\ell$. To ensure this, we assign the non-selected initial clusters to the outliers set $O$ in non-decreasing order of their cardinalities. This makes sure that we can pack maximum number of initial clusters into $O$. 

In the light of the above discussion, the problem boils down to the following Restricted Clustering problem with exactly one cluster. Let $S_j$ be the set of indices of the initial clusters colored by the color $j$ for all $1\le j\le 2B$. Also, denote the unique column in $J_j$ by $\bm{b^j}$ for $1\le j\le \beta$, and let the weight of $\bm{b^j}$, $w_j=|J_j|$. Fix a partition $T_1,T_2,\ldots,T_{\tau}$ of $\{1,2,\ldots,2B\}$. Also, guess $\tau$ positive integers $B_1,B_2,\ldots,B_{\tau}$ such that $\sum_{i=1}^{\tau} B_i=B$. $B_i$ is the guess for the cost of the $i^{th}$ composite cluster. 

Fix a part $i$ of the partition, where $1\le i\le \tau$. Let the set of color indices $T_i=\{i_1,i_2,\ldots,i_p\}$. For each $1\le t\le p$, let $U_t^i=\cup_{j\in S_{i_t}} \bm{b^j}$, i.e the set of unique columns of the initial clusters which are colored by color $i_t$. The Restricted Clustering problem that we need to solve is the following. We are given the sets of columns $U_1^i,U_2^i,\ldots,U_p^i$ and a  parameter $B_i$. The goal is to select $p$ columns $\bm{b^{j_1}},\bm{b^{j_2}},\ldots,\bm{b^{j_p}}$ and a cluster center $\bm{s}^i$ such that $\bm{b^{j_t}}\in U_t^i$ for $1\le t\le p$ and $\sum_{t=1}^{p} w_{j_t}\cdot d_H(\bm{b^{j_t}},\bm{s}^i) \le B_i$.  

%It follows from \cite{abs-1803-06102} that the above mentioned problem can be solved in $2^{O(k_i\log k_i)}{\vert\Sigma\vert}^{k_i}(nm)^{O(1)}$ time. They also show that it is possible to derandomize the algorithm using the same time complexity. Putting everything together, we obtain the following theorem. 
Next, we will show that Restricted Clustering can be solved in $2^{O(B_i\log B_i)}{\vert\Sigma\vert}^{B_i}(nm)^{O(1)}$ time using the hypergraph based technique of Marx \cite{Marx08}.

\subsection{Solving Restricted Clustering}
For simplicity, we drop the suffix $i$ from all the notations. Thus, now we are given the sets of columns $U_1,U_2,\ldots,U_p$ and a  parameter $B$. Let $\bm{c}$ be a center corresponding to an optimum feasible clustering $\mathcal{I}$ with the set  $B=\{\bm{b^{j_1}},\bm{b^{j_2}},\ldots,\bm{b^{j_p}}\}$ of chosen columns. Also, let $U=\cup_{t=1}^p U_t$. 
%Two $r$-tuples $x$ and $y$ are said to differ from each other at location $j$ if $x[j]\ne y[j]$.  

Suppose $\bm{x}$ is a vector such that 
%for $1\le j\le p$ (for $i > 1 $), $\bm{x}_j=\bm{c}_j$ and for all $p+1\le j\le r$, the column $\bm{x}_j$ is in cluster $I_j$
there are at most $B$ positions $h$ with $\bm{x}[h]\ne \bm{c}[h]$. Consider the hypergraph $H$ defined in the following way with respect to $\bm{x}$. The labels of the vertices of $H$ are in $\{1,2,\ldots,m\}$. For each $\bm{b^{j_t}}\in B$, we add $w_{j_t}$ copies of the edge $S\subseteq \{1,2,\ldots,m\}$ such that $h\in S$ if $\bm{x}[h]\ne \bm{y}[h]$. 

\begin{lemma}\label{lem:existenceofhypergraphrestricted}
Suppose the input is a \yesinstance. Consider a vector $\bm{x}$ such that there is at most $B$ positions $h$ where $\bm{x}[h]\ne \bm{c}[h]$. Also, consider the hypergraph $H$ defined in the above with respect to $\bm{x}$. There exists a subhypergraph $H_0^*(V_0^*,E_0^*)$ of $H$ with the following properties.
\begin{enumerate}
    \item $\vert V_0^*\vert \le B$.
    \item $\vert E_0^*\vert \le 200\ln B$.
    \item The indices in $V_0^*$ are the exact positions $h$ such that $\bm{x}[h]\ne \bm{c}[h]$. 
    \item The fractional cover number of $H_0^*$ is at most $4$.
\end{enumerate}
\end{lemma}

To prove the above lemma, first, we show the existence of a subhypergraph that satisfies all the properties except the second one. Let $P$ be the set of positions $h$ such that $\bm{x}[h]\ne \bm{c}[h]$. Let $H_0(V_0,E_0)$ be the subhypergraph of $H$ induced by $P$. We show that $H_0$ satisfies the first, third and fourth properties mentioned in Lemma \ref{lem:existenceofhypergraphrestricted}. 

By definition of $\bm{x}$, $P$ contains at most $B$ indices. The first property follows immediately. Next, we show that the third property also follows for $H_0$. As $V_0 \subseteq P$, in every position $h\in V_0$, $\bm{x}[h]\ne \bm{c}[h]$. It is sufficient to show that for any position $h$ with $\bm{x}[h]\ne \bm{c}[h]$, $h\in V_0$. Consider any such position $h$. Then, if there is at least one $y\in B$ such that $\bm{x}[h]\ne \bm{y}[h]$, there is an edge $e$ in $H$ that contains $h$. Now, $h \in P$ by definition. Thus, the edge $e'=e\cap P \in H_0$ also contains $h$. Hence, $h\in V_0$. In the other case, for all $y\in B$, $\bm{x}[h]= \bm{y}[h]$. However, this case does not occur. Otherwise, we can replace the symbol $\bm{c}[h]$ by $\bm{x}[h]$. The new center $\bm{c}$ would incur lesser cost, as now $\bm{c}[h]= \bm{y}[h]$ for all $\bm{y}\in B$ and previously $\bm{c}[h]\ne \bm{y}[h]$ for all $\bm{y}\in B$. Next, we show that the fourth property holds for $H_0$.  

\begin{lemma}
The fractional cover number of $H_0$ is at most $2$.
\end{lemma}

\begin{proof}
Note that the total number of edges of $H_0$ is $\tau=\sum_{t=1}^p w_{j_t}$. We claim that each vertex of $H_0$ is contained in at least  $\tau/2$  edges. 

Consider any vertex $h$ of $H_0$. Let $h$ be contained in $\delta\tau$ edges. Thus, by definition, $(1-\delta)\tau $ vectors of $B$ (with multiplicity equal to their weights) do not differ from $\bm{x}$ at location $h$. Consider replacing $\bm{c}[h]$ by $\bm{x}[h]$ at position $h$ of $\bm{c}$. Next, we analyze the change in cost of the clustering $\mathcal{I}$. Note that the cost corresponding to positions other than $h$ remains same. As $\bm{c}[h]\ne \bm{x}[h]$, previously each of the $(1-\delta)\tau$ columns were paying a cost of at least $1$ corresponding to position $h$. Now, as $\bm{c}[h]$ is replaced by $\bm{x}[h]$ and consequently $\bm{c}[h]=\bm{x}[h]$, for these columns, the cost decrease is at least $(1-\delta)\tau$. Now, the cost increase for the remaining $\delta\tau$ columns is at most $\delta\tau$.

As $\mathcal{I}$ is an optimum feasible clustering, the cost decrease must be at most the cost increase. It follows that $\delta\ge 1/2$. 
\end{proof}

So far we have proved the existence of a subhypergraph that satisfies all the properties except the second. Next, we prove the existence of a subhypergraph that satisfies all the properties. The proof of the following lemma is very similar to the proof of Lemma \ref{lem:hg-wid-2nd-property}.

\begin{lemma}\label{lem:hg-wid-2nd-property-restricted}
Let $B \ge 2$. Consider the subhypergraph $H_0$ that satisfies all the properties of Lemma \ref{lem:existenceofhypergraphrestricted} except (2). It is possible to select at most $200\ln B$ edges from $H_0$ such that the subhypergraph $H_0^*$ obtained by removing all the other edges from $H_0$ satisfies all the properties of Lemma \ref{lem:existenceofhypergraphrestricted}. 
\end{lemma}

\subsubsection{The Algorithm}

We consider all possible columns $\bm{x}\in U$ and for each of them construct a hypergraph $G$ whose vertices are in $\{1,2,\ldots,m\}$. For each column $\bm{b}_j\in U$, we add $w_j$ copies of the edge $S\subseteq \{1,2,\ldots,m\}$ such that $h\in S$ if $\bm{x}[h]\ne \bm{b}_j[h]$. For all hypergraphs $H_0$ having at most $B$ vertices and at most $200\ln B$ edges, we check if each vertex of $H_0$ is contained in at least $1/4$ fraction of the edges. If that is the case, we use the algorithm of Theorem \ref{th:hypergraph} to find every place $P'$ where $H_0$ appears in $G$ as subhypergraph. For each such set $P'$, we perform all possible $B'\le B$ edits in $\bm{x}$ at the locations in $P'$. After each such edit, we retrieve the edited vector $\bm{x}$ and construct a clustering considering $\bm{x}$ as its center in the following way. For each $1\le t\le p$, we select a vector $\bm{b}_j$ from $U_t$ such that $d_H(\bm{x},\bm{b}_j)$ is the minimum over all $\bm{b}_j\in U_t$. If the cost of this clustering is at most $B$, we output $\bm{x}$ as a candidate center. If no such vector is output as a candidate center, we return \say{NO}.   

Next, we analyze the correctness and time complexity of this algorithm. 
\subsubsection{The Analysis}
            
\begin{lemma}
The above algorithm successfully outputs $\bm{c}$ as a candidate center. 
\end{lemma}

\begin{proof}
Consider any vector $\bm{x}\in B$. Note that 
there is at most $B$ positions $h$ where $\bm{x}[h]\ne \bm{c}[h]$. Consider the hypergraph $G$ constructed by the algorithm corresponding to $\bm{x}$. Note that the hypergraph $H$ defined in Lemma \ref{lem:existenceofhypergraphrestricted} is a partial subgraph of $G$. Thus, the subhypergraph $H_0^*$ of $H$ is also a subhypergraph of $G$. As we enumerate all hypergraphs having at most $B$ vertices, at most $200\ln B$ edges and at most $4$ fractional covering number, $H_0^*$ must be considered by the algorithm. Let $P'$ be the place in $G$ where $H_0^*$ appears. By the third property of Lemma \ref{lem:existenceofhypergraphrestricted}, the locations in $P'$ are the exact positions $h$ such that $\bm{x}[h]\ne \bm{c}[h]$. It follows that an edit corresponding to $P'$ generates the vector $\bm{c}$, as $d_H(\bm{x},\bm{c})\le B$. Hence, $\bm{c}$ must be an output of the algorithm. 
\end{proof}

\subsubsection{Time Complexity}
Next, we discuss the time complexity of our algorithm. For each choice of $\bm{x}$, the hypergraph $G$ can be constructed in $n^{O(1)}$ time. The number of distinct hypergraphs $H_0$ with at most $B$ vertices and at most $200\ln B$ edges is $2^{O(B\log B)}$, since there are $2^B$ possibilities for each edge. Now we analyze the time needed for locating a particular $H_0$ in $G$. For any vector $\bm{y}\in B$, $d_H(\bm{y},\bm{c}) \le B$. By triangle inequality, $d_H(\bm{x},\bm{y})\le 2B$. Thus, the size of any edge in $H$ is at most $2B$, and we can remove any edge of $G$ of size more than $2B$. From Theorem \ref{th:hypergraph}, it follows that every occurrence of $H_0$ in $G$ can be found in $B^{O(B)}\cdot {(2B)}^{ {4B+1}}\cdot {n}^{5}\cdot { m }^2=B^{O(B)}n^{O(1)}m^2$ time. If $H_0$ appears at some place in $G$, it would take $O((B\vert \Sigma\vert)^{B})$ time to edit $\bm{x}$. Hence, in total the algorithm runs in $2^{O(B\log B)}{\vert\Sigma\vert}^{B}(nm)^{O(1)}$ time. 

Finally, Theorem~\ref{theorem:WCO} is proven. We restate it here for convenience.
\thmwco*
%\begin{theorem}
%(Restatement of Theorem \ref{theorem:WCO}) \probClusteringWCO can be solved in $2^{O(k\log k)}{\vert\Sigma\vert}^{k}(nm)^{O(1)}$ time.
%\label{theorem:WCO_reformulation}
%\end{theorem}

%% file: row-outliers.tex
\section{Clustering with Row Outliers}
\label{sec:rowoutliers}
Here we show that \probClusteringWRO is a special case of \probconstrainedcl. 

%Note that for any matrix, a set of row outliers can be treated as column outliers of the transposed matrix. Similarly, a clustering of the columns of the transposed matrix induces a clustering of the rows of the original matrix. However, we want to find a clustering of the columns of the original matrix. To do this, one can add extra constraints on the clustering of the transposed matrix so that a clustering of the columns implicitly induces a clustering of the rows as well of the transposed matrix. This latter clustering induces a clustering of the columns of the original matrix and is our desired clustering.     

%\begin{lemma}\label{lem:ROsplcase}
%For any instance $(\bm{D},k,B,\ell)$ of \probClusteringWRO, one can construct in time $\OO(mn+k\cdot |\Sigma|^k)$ an equivalent instance $(\bm{A},k',B',\ell',\R)$ of \probconstrainedcl such that $\bm{A}$ is the transpose of $\bm{D}$, $k'={\vert\Sigma\vert}^{k}$, $B'=B$ and $\ell'=\ell$. 
%\end{lemma}
\lemROsplcase*

\begin{proof}
Given an instance $I=(\bm{D},k,B,\ell)$ of \probClusteringWRO, we construct an instance $I'=(\bm{A},k',B',\ell',\R)$ of \probconstrainedcl such that $\bm{A}$ is the transpose of $\bm{D}$, $k'={\vert\Sigma\vert}^{k}$, $B'=B$ and $\ell'=\ell$. The set of constraints $\R$ is defined as follows. Let $S=\{s_1,s_2,\ldots,s_{{\vert\Sigma\vert}^k}\}$ be the lexicographic ordered collection of all ${\vert\Sigma\vert}^k$ strings of length $k$ on alphabet $\Sigma$. Also, let $Q$ be the ${\vert\Sigma\vert}^k\times k$ matrix such that the $i$-th row of $Q$ is the $i$-{th} string of $S$. Thus, each element of $Q$ is in $\Sigma$. For each column $\bm{q}^j$ of $Q$ with $1\le j\le k$, we construct a tuple $z_j=(\bm{q}^j[1],\bm{q}^j[2],\ldots,\bm{q}^j[{\vert\Sigma\vert}^k])$. Let $R$ be the set of these tuples. To construct $\R$ we set $R_i=R$ for all $i$. This finishes our construction.  

Next, we prove that $I$ is a \yesinstance of \probClusteringWRO if and only if $I'$ is a \yesinstance of \probconstrainedcl. The lemma follows immediately. 

First, suppose $I'$ is a \yesinstance of \probconstrainedcl. Let $X=\{X_1,X_2,\ldots$ $,X_{{\vert\Sigma\vert}^{k}}\}$ be a feasible clustering of the columns of $\bm{A}\setminus O$, the matrix $\bm{A}$ after removing the columns in $O$, where $O$ is the set of outliers. Also, let $\bm{c}_1,\bm{c}_2,\ldots,\bm{c}_{{\vert\Sigma\vert}^k}$ be the corresponding centers. Note that as this clustering satisfies the constraints in $\R$, each tuple $(\bm{c}_1[j],\bm{c}_2[j],\ldots,\bm{c}_{{\vert\Sigma\vert}^k}[j])$ is one of the $k$ tuples of $R$. Let $Y=\{Y_1,Y_2,\ldots,Y_{k}\}$ be the  partition of the rows of $\bm{A}\setminus O$, such that $Y_t$ contains all the indices $j$ such that $(\bm{c}_1[j],\bm{c}_2[j],\ldots,\bm{c}_{{\vert\Sigma\vert}^k}[j])=z_t$ for $1\le t\le k$. We reorder the rows of matrix $\bm{A}\setminus O$ based on this partition. In particular, we modify the row numbers in a way so that the rows in $Y_1$ become the first $\vert Y_1\vert$ rows, the rows in $Y_2$ become the next $\vert Y_2\vert$ rows, and so on (see Figure \ref{fig:matrix}). Now, let us go back to the clustering $\{X_1,X_2,\ldots,X_{{\vert\Sigma\vert}^k}\}$. Note that this is still a clustering of the new matrix $\bm{A}\setminus O$ with cost at most $B$ with rows of the centers reordered in the same way. Indeed, by construction, the reordering of the centers ensures that the center of $X_i$ is 
the column vector $(z_1[i],\ldots \vert Y_1\vert$ times $,z_2[i],\ldots \vert Y_2\vert$ times $,\ldots, z_k[i],\ldots \vert Y_k\vert$ times$)$ (see Figure \ref{fig:matrix}). Now, note that $\{Y_1,Y_2,\ldots,Y_{k}\}$ is a clustering of the columns in $\bm{D}^{-O}$. We claim that its cost is at most $B$. As $\vert O\vert\le l$, it follows that $I$ is a \yesinstance of \probClusteringWRO. Note that some clusters might be empty. But, as they do not incur any cost, we keep them for simplicity. To analyze the cost, set the center of $Y_j$ to the column vector $(z_j[1],\ldots \vert X_1\vert$ times $,z_j[2],\ldots \vert X_2\vert$ times $,\ldots, z_j[{\vert\Sigma\vert}^k],\ldots \vert X_{{\vert\Sigma\vert}^k}\vert$ times$)$ (see Figure \ref{fig:matrix}). Now, because of the symmetry between the two clusterings $X$ and $Y$, their costs are the same. Hence, the claim follows. 

Now, suppose $I$ is a \yesinstance of \probClusteringWRO and consider a feasible clustering $Y=\{Y_1,Y_2,\ldots,Y_{k}\}$. Let $O$ be the set of indices of the outliers. Also, let $\bm{c}_1,\bm{c}_2,\ldots,\bm{c}_k$ be the corresponding centers. Note that each string of the form $\bm{c}_1[j]\bm{c}_2[j]\ldots\bm{c}_k[j]\in S$. For a string $s \in S$, we say row $j$ is of type $s$ if $\bm{c}_1[j]\bm{c}_2[j]\ldots\bm{c}_k[j]=s$. Let $\{X_1,X_2,\ldots,X_{{\vert\Sigma\vert}^k}\}$ be the partition of the rows of $\bm{D}^{-O}$ based on their types. In particular, for the $i$-th string $s\in S$, $X_i$ contains all rows of type $s$. Note that $\{X_1,X_2,\ldots,X_{{\vert\Sigma\vert}^k}\}$ is also a clustering of the columns of $\bm{A}\setminus O$. We show that the cost of this clustering is at most $B$. To analyze the cost, choose the column vector $(s_i[1],\ldots \vert Y_1\vert$ times $,s_i[2],\ldots \vert Y_2\vert$ times $,\ldots, s_i[k],\ldots \vert Y_k\vert$ times$)$ as the center vector of the cluster $X_i$ of the columns of $\bm{A}\setminus O$. Note that the centers chosen in this way satisfy the constraints in $\R$, as the center tuple $(s_1[j],s_2[j],\ldots,s_{{\vert\Sigma\vert}^k}[j])=(\bm{q}^j[1],\bm{q}^j[2],\ldots,\bm{q}^j[{\vert\Sigma\vert}^k])$ corresponding to index $j$ is in $R$. It follows that the cost of this clustering is at most $B$, as it is induced by the clustering $Y$ and the centers $\bm{c}_1,\bm{c}_2,\ldots,\bm{c}_k$. 
\end{proof}

\begin{figure}[ht]
\centering
\includegraphics[scale=.7]{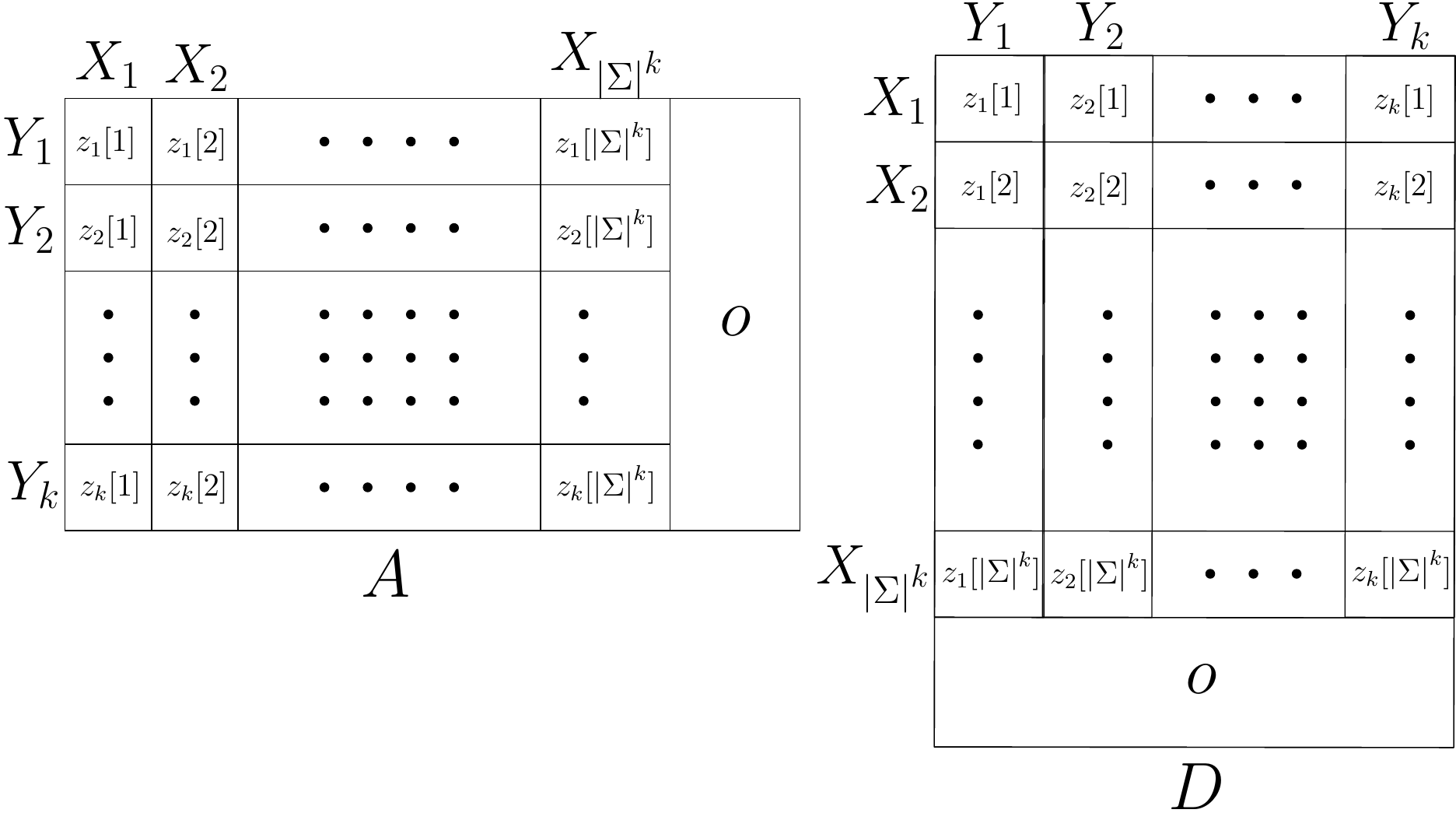}
\caption{Figure showing the clustering $X$ of $\bm{A}\setminus O$ and $Y$ of $\bm{D}^{-O}$. The common center coordinate values are shown inside submatrices.}
\label{fig:matrix}
\end{figure}

Lemma~\ref{lem:ROsplcase} together with Theorem~\ref{theorem:main} immediately give us Theorem~\ref{theorem:clustering_outliers} with $f(k,B,\vert \Sigma\vert)= (k'B)^{O(k'B)}{\vert \Sigma\vert}^{k'B}$ and $g(k,\vert \Sigma\vert)=O(k')$, where $k'={\vert \Sigma\vert}^{k}$. The theorem is restated here.
\thmwro*
%\begin{corollary}
%(Restatement of theorem \ref{theorem:main}) \probClusteringWRO is FPT parameterized by $k$ if $r$ and $|\Sigma|$ are constants. 
%\label{corollary:WRO}
%\end{corollary}

%% file: hardness.tex
%TODO
% check notation consistency, esp. graph notation
\section{Lower bounds}\label{sec:hardness}

In this section we present the lower bound results. First we recall a few complexity notions essential to our hardness proofs.

%\medskip\noindent\textbf{Complexity.}
%\todo{where should we put it?}
 A \emph{parameterized problem} is a language $Q\subseteq \Sigma^*\times\mathbb{N}$ where $\Sigma^*$ is the set of strings over a finite alphabet $\Sigma$. Respectively, an input  of $Q$ is a pair $(I,k)$ where $I\subseteq \Sigma^*$ and $k\in\mathbb{N}$; $k$ is the \emph{parameter} of  the problem. 
A parameterized problem $Q$ is \emph{fixed-parameter tractable} (\classFPT) if it can be decided whether $(I,k)\in Q$ in time $f(k)\cdot|I|^{\Oh(1)}$ for some function $f$ that depends of the parameter $k$ only. Respectively, the parameterized complexity class \classFPT is composed by  fixed-parameter tractable problems.
The $\operatorClassW$-hierarchy is a collection of computational complexity classes: we omit the technical
definitions here. The following relation is known amongst the classes in the $\operatorClassW$-hierarchy:
$\classFPT=\classW{0}\subseteq \classW{1}\subseteq \classW{2}\subseteq \ldots \subseteq \classW{P}$. It is widely believed that $\classFPT\neq \classW{1}$, and hence if a
problem is hard for the class $\classW{i}$ (for any $i\geq 1$) then it is considered to be fixed-parameter intractable.
We refer to   books \cite{CyganFKLMPPS15,DowneyF13} for the detailed introduction to parameterized complexity.  

We also make use of  the following complexity hypothesis formulated by  Impagliazzo, Paturi, and Zane   \cite{ImpagliazzoPZ01}.
 
\begin{quote}
\textbf{Exponential Time Hypothesis (ETH)}:  There is a positive real $s$ such that 3-CNF-SAT with $n$ variables and $m$ clauses cannot be solved in time $2^{sn}(n+m)^{\Oh(1)}$.
\end{quote}

Our lower bounds hold for the \probClusteringWRO problem, and even for its special case where $B = 0$. This case may be viewed upon as the problem of partitioning columns of the input matrix into clusters where columns in each cluster are identical, except for the outlier rows. We show that this exact version of \probClusteringWRO is W[1]-hard when parameterized by the number of clusters $k$ and the number of inliers $(n - \ell)$, and also W[1]-hard when parameterized by the number of outliers $\ell$.

The first lower bound shows that the result in Theorem~\ref{theorem:clustering_outliers} is tight in the sense that $k$ must be a constant. This is in contrast to the \probClusteringWCO problem, for which we show the FPT algorithm for the stronger parameterization (Theorem~\ref{theorem:WCO}). Thus, \probClusteringWRO is fundamentally more difficult than \probClusteringWCO. Also, note that \probClusteringWCO is trivially solvable in polynomial time when the cost $B$ is zero.

\begin{lemma}
    \probClusteringWRO is W[1]-hard when parameterized by $k + (m - \ell)$ when $B = 0$ and $\Sigma=\{0,1\}$. Assuming ETH, there is no algorithm solving the problem with $B = 0$ and the binary alphabet in time $m^{o(k)}\cdot n^{\Oh(1)}$.
    \label{lemma:w1r}
\end{lemma}
\begin{proof}
        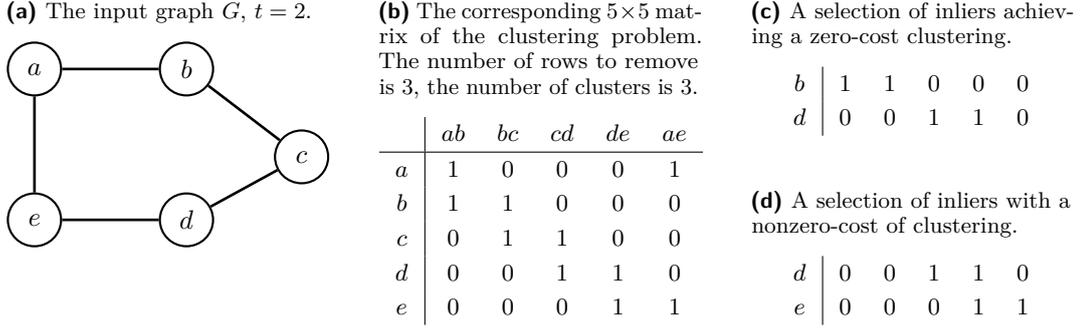
\begin{figure}[ht]
            \centering
            \tabskip=0pt
            \valign{#\cr
                \hbox{%
            \begin{subfigure}[b]{.3\textwidth}
                \centering
                \caption{The input graph $G$, $t = 2$.}
                \begin{tikzpicture}[auto,  node distance=2cm,every loop/.style={},thick,main node/.style={circle,draw, minimum size=.7cm}]
                    \node[main node] (1) {$a$};
                    \node[main node] (2) [right of=1] {$b$};
                    \node[main node] (3) [below right=.65cm and 1cm of 2] {$c$};
                    \node[main node] (4) [below of =2]{$d$};
                    \node[main node] (5) [below of=1] {$e$};
                    \draw[line width=1pt]
                    (1) --  (2);
                    \draw[line width=1pt]
                    (2) -- (3);
                    \draw[line width=1pt]
                    (3) -- (4);
                    \draw[line width=1pt]
                    (4) -- (5);
                    \draw[line width=1pt]
                    (5) -- (1);
                \end{tikzpicture}
            \end{subfigure}%
        }\cr
        \noalign{\hfill}
        \hbox{%
            \begin{subfigure}[b]{.3\textwidth}
                \centering
                \caption{The corresponding $5 \times 5$ matrix of the clustering problem. The number of rows to remove is $3$, the number of clusters is $3$.}
                \begin{tabular}{c | c c c c c}
                    &$ab$&$bc$&$cd$&$de$&$ae$\\\hline
                    $a$&1&0&0&0&1\\
                    $b$&1&1&0&0&0\\
                    $c$&0&1&1&0&0\\
                    $d$&0&0&1&1&0\\
                    $e$&0&0&0&1&1\\
                \end{tabular}
            \end{subfigure}%
        }\cr
        \noalign{\hfill}
        \hbox{%
            \begin{subfigure}[b]{.3\textwidth}
                \centering
                \caption{A selection of inliers achieving a zero-cost clustering.}
                \begin{tabular}{c | c c c c c}
                    $b$&1&1&0&0&0\\
                    $d$&0&0&1&1&0\\
                \end{tabular}
            \end{subfigure}%
        }\vfill
        \hbox{%

            \begin{subfigure}[b]{.3\textwidth}
                \centering
                \caption{A selection of inliers with a nonzero-cost of clustering.}
                \begin{tabular}{c | c c c c c}
                    $d$&0&0&1&1&0\\
                    $e$&0&0&0&1&1\\
                \end{tabular}
            \end{subfigure}%
        }\cr
    }
            \caption{An illustration of the reduction in Theorem~\ref{thm:w1r} showing a possible input graph, the obtained matrix, and two possible selections of the inliers.}
            \label{fig:w1r}
        \end{figure}

        We show a reduction from the \probIS problem. Given a graph $G(V, E)$ and a parameter $t$, the task in \probIS is to determine whether there is an independent set of size at least $t$ in $G$. Consider an instance $(G(V, E), t)$ of \probIS. Assume that for any $t$-sized subset $S$ of $V$, each of $s$ in $S$ is adjacent to a vertex in $V \setminus S$, and there is an edge inside $V \setminus S$. Any graph $G$ can be tweaked to meet this assumption by adding a $(t + 2)$-sized clique $C$ to the graph and connecting each vertex of $G$ to each vertex of $C$. Note that no vertex of $C$ can be in any independent set of size at least two, and any independent set of the original graph is present in the newly constructed graph. Thus, the new instance of \probIS is equivalent to the original one, and it satisfies the assumption stated.
    
    From $(G, t)$, we construct an instance of \probClusteringWRO over the alphabet $\Sigma = \{0, 1\}$. The input matrix $\bm{A}$ is the incidence matrix of $G$ where rows are indexed by vertices and columns by edges. Set $k = t + 1$, $\ell = |V| - t$, $B = 0$. See Figure~\ref{fig:w1r} for an example.

    Now it remains to show the correctness of the reduction. Assume there is an independent set $I$ of size $t$ in $G$. Let the inliers be the rows corresponding to $I$, and $O$ be the set of the remaining outlier rows.
    Restricted to the inlier rows, every column of $\bm{A}$ is either a zero column, or has a single one, since no edge of $G$ has both ends in $I$.
    Thus, there are at most $t + 1$ distinct columns in $\bm{A}^{-O}$, so there is a zero-cost $(t + 1)$-clustering. The centroids in this clustering are the following $t$-columns: the zero column and all the $t$ possible columns with a single one.

    For the other direction of the correctness proof, assume there is a solution to the constructed \probClusteringWRO instance. Consider the set $I$ of the vertices corresponding to the inlier rows, and the set $O$ of the outlier rows.
    By the assumption on $G$, every vertex in $I$ is adjacent to a vertex in $V \setminus I$, and there is an edge inside $V \setminus I$. Thus, $\bm{A}^{-O}$ has a zero column and all $t$ distinct columns with a single one. If there is an edge $e$ inside $I$, then there are at least $t + 2$ distinct columns in the matrix, since even after removing the rows of $O$, the column corresponding to $e$ has two ones. In this case no zero-cost clustering with $t + 1$ clusters is possible. Thus, $I$ must be an independent set.

    Thus, the reduction is valid. It is well-known that \probIS is W[1]-hard, and also not solvable in time $f(t) \cdot |V|^{o(t)}$ for any computable function $f$ unless ETH fails.
    Observing that $k = t + 1$ and $m - \ell = t$ concludes the proof of the lemma.
\end{proof}

Since Lemma~\ref{lem:ROsplcase} gives a parameterized reduction from \probClusteringWRO to \probconstrainedcl, we immediately get the following corollary.
\begin{corollary}
    \probconstrainedcl is W[1]-hard when parameterized by $k + (n - l)$ when $B = 0$ and $\Sigma=\{0,1\}$.
\end{corollary}
We do not get the analogous ETH bound however, as the number of clusters in the constructed instance of \probconstrainedcl is exponential.

The second lower bound shows that bounding both the number of outliers and the cost $B$ is insufficient for an FPT algorithm as well.

\begin{lemma}
    \probClusteringWRO is W[1]-hard when parameterized by $\ell$ when $B = 0$ and $\Sigma=\{0,1\}$. Assuming ETH, there is no algorithm solving the problem with $B = 0$ and the binary alphabet in time $m^{o(l)}\cdot n^{\Oh(1)}$.
    \label{lemma:w1l}
\end{lemma}
\begin{proof}
    We show a reduction from the \probPVC problem.
    The input to \probPVC is a graph $G(V, E)$ and numbers $t$, $q$. The problem is to decide whether there is a subset $C \subset V$ of size at most $t$ such that at least $q$ edges of $G$ are covered by $C$.
    The \probPVC problem is well-known to be as hard as \probIS (\cite{CyganFKLMPPS15}, Theorem 13.6). That is, when \probPVC is parameterized by $t$, there is a parameter-preserving reduction from \probIS.

    Consider an instance $(G, t, q)$ of \probPVC. The idea of the reduction is similar to that in Theorem~\ref{thm:w1r}. In short, we argue that choosing $t$ vertices in $G$ in the way that maximizes the number of covered edges corresponds to removing $t$ rows in the incidence matrix of $G$ in the way that minimizes the number of distinct columns, up to a certain assumption on $G$. Next, we show that formally.

    First, we modify $G$ to obtain a new graph $G'(V', E')$. Let $P$ be a set of two new vertices, and $D$ be a set of $d$ new vertices, where $d = 5 + t + |E| - q$.
    Set $V'$ to $V \cup P \cup D$. Keep all edges of $G$ on $V$, and connect each vertex of $P$ to all other vertices of $G'$, including the other vertex of $P$. That is,
    \[E' = E \cup \left\{pv\ |\ p \in P, v \in V'\setminus\{p\}\right\}.\]
    The total number of edges in $G'$ is $|E| + 2 \cdot |V'| - 3$.

    The graph $G'$ behaves in the same way as $G$ with respect to \probPVC, as stated in the following claim. Set $t' = t + 2$ and $q' = q + 2 \cdot|V'| - 3$.
    \begin{claim}
        The instance $(G, t, q)$ is a yes-instance of \probPVC if and only if $(G', t', q')$ is a yes-instance of \probPVC.
    \end{claim}
    \begin{proof}
        First, assume there is a subset $C \subset V$ of size at most $t$ covering at least $q$ edges of $G$. Set $C' = C \cup P$, $|C'| \le t + 2 = t'$, and vertices of $P$ cover the additional $2 \cdot |V'| - 3$ edges which are not present in $G$.

        To the other direction, assume there is a subset $C' \subset V'$ of size at most $t'$ covering at least $q'$ edges of $G'$. We may assume that $P \subset C'$, otherwise we may replace any vertex of $C$ by a vertex of $P$, and the number of covered edges does not decrease since vertices in $P$ are adjacent to all vertices. Now, $C = C' \cap V$ is a solution for $(G, t, q)$: $|C| \le |C' \setminus P| \le t$, and since
        $|E' \setminus E| = q' - q$, $C$ must cover at least $q$ edges in $G$.
    \end{proof}

    Now we construct the input instance $(\bm{A}, k, B, \ell)$ of \probClusteringWRO over the alphabet $\Sigma = \{0, 1\}$.
    The input matrix $\bm{A}$ is the incidence matrix of $G'$ where rows are indexed by vertices and columns by edges. Set $k = 1 + |V'| - t' + |E'| - q'$, $\ell = t'$, $B = 0$.

    To show the correctness of the reduction, first assume there is a set $C \subset V'$ of size $t'$ covering at least $q'$ edges in $G'$. Let $O$ be the set of rows in $\bm{A}$ corresponding to $C$. We claim that there are at most $1 + (|V'| - t') + (|E'| - q')$ distinct columns in $\bm{A}^{-O}$. That is, at most one zero column, at most $|V'| - t'$ columns with a single one corresponding to the vertices of $V' \setminus C$ adjacent to $C$, and at most $|E'| - q'$ columns with two ones corresponding to the edges not covered by $C$. Since there are at most $k$ distinct columns, the cost of $k$-clustering is zero if we remove the rows of $O$.

    In the other direction, assume there is a solution to the constructed \probClusteringWRO instance, that is, a set $O$ of at most $t'$ rows such that $\bm{A}^{-O}$ has at most $k$ distinct columns. Consider the set $C$ of vertices corresponding to the outlier rows $O$.
First, we show that $C$ must contain $P$. If this does not hold, either $C \cap P$ is empty, or $|C \cap P| = 1$. Assume $C \cap P$ is empty, in this case out of $2 \cdot |V'| - 3$ edges incident to $P$ at most $2t'$ are covered by vertices of $C$. These edges correspond to at least $2 \cdot |V'| - 3 - 2t'$ columns of $\bm{A}$ that have still two ones each in $\bm{A}^{-O}$, and thus they are all distinct. However, the number of clusters is smaller than the number of such columns, since
    \[2 \cdot |V'| - 3 - 2 t' = |V'| - 2t' - 1 + |V| + d > 1 + |V'| - t' + |E'| - q' = k,\]
    as $d = 3 + t' + |E| - q$ and $|E| - q = |E'| - q'$. This is a contradiction to $O$ being a solution.

    In the other case $|C \cap P| = 1$, denote by $p$ the vertex of $P$ that is in $C$, and by $p'$ the other vertex of $P$. At most $t'$ vertices in $V' \setminus \{p\}$ are in $C$, thus there are at least $|V'| - 1 - t'$ distinct columns in $\bm{A}^{-O}$ that correspond to edges from $p$ to $V' \setminus C$, they each have a single one in the row corresponding to the second endpoint.
    There are also at least $|V'| - 2 - t'$ distinct columns in $\bm{A}^{-O}$ that correspond to edges from $p'$ to $V' \setminus C$, they each have two ones in the rows corresponding to $p'$ and the second endpoint. Thus, there are at least $2 \cdot |V'| - 3 - 2t'$ distinct columns in $\bm{A}^{-O}$, and we have already shown that this quantity is strictly larger than $k$, leading to the contradiction.

    Now that we have shown $P \subset C$, we argue that all possible columns with less than two ones are present in $\bm{A}^{-O}$, giving an exact bound on the number of distinct columns with two ones.
    Formally, since $P \subset C$, there is a zero column in $\bm{A}^{-O}$ corresponding to the edge inside $P$. Take any $p \in P$, since $p \in C$ and any vertex $v$ in $|V'| \setminus C$ is adjacent to $p$,
    there is a column in $\bm{A}^{-O}$ corresponding to the edge $vp$ that has a single one in the row corresponding to $v$.
    This gives $|V'| - t'$ distinct columns with single ones, one for each vertex of $V' \setminus C$. Thus, there must be at most $k - 1 - (|V'| - t') = |E'| - q'$ distinct columns with two ones in $\bm{A}^{-O}$. These columns correspond exactly to the uncovered by $C$ edges of $G'$. So $C$ covers at least $q'$ edges, and $(G', t', q')$ is a yes-instance. This finishes the correctness proof, and from the hardness of \probPVC the lemma follows.
\end{proof}
 
Clearly, Theorem~\ref{thm:w1r}, restated next for convenience, follows from Lemma~\ref{lemma:w1r} together with Lemma~\ref{lemma:w1l}.
\thmhardness*

%% file: conclusion.tex
 %!TEX root = column-outliers.tex
 
%%%%%%%%%%%%%%

\section{Conclusion}\label{sec:concl}

We initiated the systematic study of parameterized complexity of robust categorical data clustering problems. In particular, for \probClusteringWCO, we proved that the problem can be solved in $2^{\Oh(B\log B)}|\Sigma|^B\cdot (nm)^{\Oh(1)}$ time. Further, we considered the case of row outliers and proved that \probClusteringWRO  is solvable  in time 
$f(k,B,|\Sigma|)\cdot m^{g(k,|\Sigma|)}n^2$. We also proved that  we cannot avoid the dependence on $k$ in the degree of the polynomial of the input size in the running time unless $\classW{1}=\classFPT$, and  the problem cannot be solved in $m^{o(k)}\cdot n^{\Oh(1)}$ time, unless ETH is false.  
%
%,  it is natural to ask whether this dependence on $r$ could be made polynomial (or even linear) or the exponent $2^r$ cannot be improved.  Note that our current lover bound shows only that the problem cannot be solved  in time $n^{o(r)}\cdot m^{\Oh(1)}$ up to ETH for $k=0$ and $\Sigma=\{0,1\}$.
%
To deal with row outliers, we introduced the \probconstrainedcl problem and obtained the algorithm with running time $(kB)^{\Oh(kB)}|\Sigma|^{kB}\cdot m^2n^{\Oh(k)}$. This problem is very general, and the algorithm for it not only allowed us to get the result for \probClusteringWRO, but also led to the algorithms   
for the robust low rank approximation problems. In particular, we obtained that  \probRLRMA is \classFPT if $k$ and $p$ are constants when the problem is parameterized by  $B$. However, even if the low rank approximation problems are closely related to the matrix clustering problems,
there are structural differences. For instance, we show that the complexity of clustering with column outliers is different from row outliers, however, low-rank approximation problems are symmetric.
%they are quite different.  Just to give an illustrative example, the parameterized complexity of the column and row outliers variants of the clustering problem considered in our paper are different with respect to $k$, but clearly the low rank approximation problems for matrices over fields  are symmetric with respect to rows and columns.
This leads to the question whether  \probRLRMA,  \probRBRMA and  \probProjClR could be solved by better algorithms specially tailored for these problems. It is unlikely that potential improvements would considerably change the general qualitative picture. For example, \probRLRMA for $p=2$ and $\ell=0$ is \classNP-complete if $k=2$~\cite{DanHJWZ15,GillisV15} and \classW{1}-hard when parameterized by $B$~\cite{FominLMSZ18}. It is also easy to observe that  \probRLRMA for $p=2$, $B=0$ and $k=n-\ell-1$ is equivalent to asking whether the input matrix $\bm{A}$ has $n-\ell$ linearly dependent columns. This immediately implies that  \probRLRMA for $p=2$ and $B=0$ is \classW{1}-hard when parameterized by $k$ or $n-\ell$ by the recent results about the \textsc{Even Set} problem~\cite{BhattacharyyaBEGKBMM19}. 
%The open case is the complexity  of \probRLRMA  parameterized by $p$ and $B$ when $k$ is a constant.  
%However,  t
The most interesting open question,  by our opinion,  is whether the exponential dependence on $k$ in the degree of the polynomial of the input size in the running time
produced by our reduction of  \probRLRMA to \probconstrainedcl %from Lemma~\ref{lem:BmatrixFasRClust}
could be avoided, even if $p$ is a constant. Can the dependence of $k$ be  made polynomial (or even linear)?